\title{On the Characterization of $1$-sided error Strongly-Testable Graph
  Properties for bounded-degree graphs, including an appendix\\
}
\author
{Hiro Ito\thanks{School of Informatics and Engineering, 
Univ. of Electro-Communications and 
Crest, JST, Japan, 
email: itohiro@uec.ac.jp}\hspace{0.5cm}
	Areej Khoury\thanks{
		Department of Computer Science, 
		University of Haifa, Haifa, Israel, 
		email: areej.khoury@gmail.com
	} \hspace{0.5cm}
	Ilan Newman\thanks{
		Department of Computer Science, 
		University of Haifa, Haifa, Israel, 
		email: ilan@cs.haifa.ac.il. This research was
                supported by The Israel Science Foundation, grant number
497/17.} 
	}

\date{\today}
\documentclass[11pt]{article}
\usepackage{amssymb,fullpage}
\usepackage{xspace}
\usepackage{latexsym}
\usepackage{times}
\usepackage{amsfonts}
\usepackage{amsmath}

\usepackage{comment}

\usepackage{algorithm}
\usepackage{algpseudocode}
\usepackage{amsmath}

\algnewcommand\algorithmicforeach{\textbf{for each}}
\algdef{S}[FOR]{ForEach}[1]{\algorithmicforeach\ #1\ \algorithmicdo}

\def \e {\epsilon}
        
\def \Sink{RV}

\begin{document}
\maketitle
\begin{abstract}
  We study property testing of (di)graph properties in bounded-degree
  graph models.  The study of graph properties in bounded-degree
  models is one of the focal directions of research in property
  testing in the last 15 years. However, despite of the many results
  and the extensive research effort, there is no characterization of
  the properties that are strongly-testable (i.e., testable with constant query
  complexity) 
  even for $1$-sided error tests.

	The bounded-degree model can naturally be generalized to
        directed graphs resulting in two models that were
        considered in the literature. The first contains the directed
        graphs in which the outdegree is bounded but the indegree is
        not restricted. In the other, both the outdegree and indegree
        are bounded.

		We give a characterization of the $1$-sided error
strongly-testable {\em monotone} graph properties, and the $1$-sided
error strongly-testable {\em hereditary} graph properties in all the
bounded-degree directed and undirected graphs models.

{\bf comments: this version corrects minor details in the previous: (a)
   removed the non-defined term 'non-redundant' from
  theorem 3.3.  (b) corrected a typo in example 7.2 page 22} 
\end{abstract}

\newpage
\clearpage
\setcounter{page}{1}
\newcommand{\ignore}[1]{}
\def \qed {\hspace*{0pt} \hfill {\quad \vrule height 1ex width 1ex depth 0pt}
\medskip}

\newenvironment{proof}{\par\noindent{\bf Proof.}\quad}{  $\qed$}

\newcommand{\R}{\ensuremath{\mathbb R}}
\newcommand{\Z}{\ensuremath{\mathbb Z}}
\newcommand{\N}{\ensuremath{\mathbb N}}
\newcommand{\F}{\ensuremath{\mathcal F}}

\newtheorem{theorem}{Theorem}[section]
\newtheorem{definition}{Definition}[section]
\newtheorem{claim}[definition]{Claim}
\newtheorem{proposition}{Proposition}[section]
\newtheorem{observation}{Observation}[section]
\newtheorem{remark}[definition]{Remark}
\newtheorem{prop}[section]{Proposition}
\newtheorem{lemma}[proposition]{Lemma}
\newtheorem{crl}{Corollary}
\newtheorem{corol}{Corollary}[section]
\newcommand{\Proof}{\noindent{\bf Proof.}\ \ }
\newtheorem{corollary}{Corollary}[section]
\newtheorem{fact}[definition]{Fact}
\newtheorem{conj}[theorem]{Conjecture}
\newtheorem{thm}[definition]{Theorem}
\newtheorem{thmm}[proposition]{Theorem}
\newtheorem{thmp}[proposition]{Theorem}
\newtheorem{remarkp}[proposition]{Remark}
\newtheorem{definitionp}[proposition]{Definition}
\newtheorem{claimp}[proposition]{Claim}
\newtheorem{factp}[proposition]{Fact}

\bibliographystyle{plain}


\section{Introduction}

Testing graph properties has been at the core of
combinatorial property testing since the very beginning with the
important results of Goldreich-Goldwasser-Ron \cite{GGR98}. There are several  different models
of interest.
 In the dense graph model an
$n$-vertex graph is given by its $n \times n$ boolean adjacency
matrix. For  this model  there are
characterizations of the properties that can be tested in constant
amount of queries by $1$-sided error tests \cite{AS08}, $2$-sided
error tests \cite{AFNS09}, and the properties that are defined by
forbidden induced  subgraphs and are testable by very small query complexity
\cite{AS06}.

In the other model,  called the {\em incidence-list} model,   an
$n$-vertex graph is represented  by its incidence lists. That is, an
array of size $n$ in which every entry is associated with a vertex,
and contains a list of the neighbours of that vertex. 
This model contains the important special case of the {\em bounded-degree model} in which  the degree of the vertices 
is bounded by a universal parameter $d$ (and hence the lists are of size at most $d$).

The bounded-degree model, first considered in the property testing context
by Goldreich and Ron \cite{GR02},
attracts much of the research interest in combinatorial property
testing in the past decade. One reason is  the algorithmic sophistication and wealth of
structural results that were developed in the studies of property testing
in this model. 
E.g., the use of random walks to test partition properties, starting in
\cite{bounded-degree}, and with the sophisticated recent results in
\cite{CS10, CPS15} for
expander and clustering testing, the ``local-partition'' oracle
\cite{HKNO09, LR15, NO08}, and
others.
 The other motivation is the rapidly growing research of very large
networks, e.g.,  the Internet, and other natural large 
networks such as social networks.
These large networks often turn to be represented by bounded-degree
(di)graphs (or very sparse (di)graphs). 
 Property testing of sparse
 graphs can provide a useful filter to discard unwanted instances  at a very low cost (in time
and space), as well as algorithmic and structural insights regarding the
tested properties.

Despite of the focus and wealth of results,  the
bounded-degree  model  remains far from being understood. In particular,   as of
present, there is no characterization of the 
properties that are testable in constant query complexity,
neither by $2$-sided error tests, nor by $1$-sided error tests. 

We focus on $1$-sided error testing. Our main result is a
characterization of the {\em monotone} (di)graph properties, and the
{\em hereditary} (di)graph properties, that are $1$-sided-error {\em
  strongly-testable}\footnote{For formal
  definiton of ``property testing'' see Section \ref{sec:prelim}}.  Here ``strongly-testable'' means that the
property can be tested by a constant number of queries that is
independent of the graph size, but may depend on the distance
parameter $\epsilon$.
The characterization essentially states that a monotone graph property
is strongly-testable if and only if it is close (see Definition
\ref{def:dist}) to a property that is defined by a set of forbidden
subgraphs of constant size  (Theorem \ref{thm:main-fb}). For
hereditary property we obtain a similar result  (Theorem
\ref{thm:main-fb0.5}) except that 
forbidden subgraphs are replaced with forbidden as induced subgraphs.

We
believe that our results form a first step towards a characterization
of all $1$-sided error strongly-testable graph properties in the bounded-degree model.
 
The bounded-degree model extends naturally to directed graphs. There
are two different models that have been studied for directed
graphs: In the first, the access to the graph is via queries to
outgoing neighbours, and correspondingly, only the out-degree of
vertices is bounded. This model
corresponds to the standard representation of directed graphs in
algorithmic computer science. Namely, where an $n$-vertex directed graph (digraph) is
represented by $n$ lists, each being associated with a distinct vertex $v$
in the graph, and contains the list of {\em forward} edges going out
from $v$.  
The access to a $d$-outdegree bounded digraph in this model is via queries of the following
type: a query specifies  a pair
$(v,i)$ where $v \in V(G)$ and $i \leq d$. As a response, the
algorithm discovers  the $i$th outgoing neighbour of the vertex $v$\footnote{if there is
one, or a special symbol otherwise}.  In what follows we abbreviate this model
as the $F(d)$-model, where $d$ is the upper bound on the out-degree of
vertices.

In the other model,  both the  in-degree
and out-degree are bounded by $d$. In this case  an $n$-vertex  graph is represented
by $2n$ lists; the list of outgoing edges and the list of incoming edges for
each vertex. The query type changes accordingly and allows both 
`outgoing' and `incoming' edge queries.  We denote this model 
as the $FB(d)$-model (`forward' and ` backward' queries). This model contains  the model of undirected $d$-bounded
degree graphs (where each undirected edge is replaced by a pair of
anti-parallel edges).

We note that the $F(d)$ model, as a collection
of graphs, strictly contains the $FB(d)$ model, while algorithmically
it is more restricted by the limited access to the graph. 

In all models, an $n$-vertex  (di)graph $G$ is said to be
$\epsilon$-far from a (di)graph property $P$ if it is
required to change (delete and/or insert) at least $\epsilon \cdot dn$
edges in order to get a $d$-bounded degree graph (in the corresponding
model) that has the property $P$.  

 The  results in this paper are the characterization of
the {\em monotone}  digraph properties and {\em hereditary} digraph
properties 
that are $1$-{\em sided error} strongly-testable in the $F(d)$-model
(Theorem \ref{thm:main-mon}, and 
\ref{thm:main-hered}).  The results for the $FB(d)$ model easily follow from
these for the $F(d)$-model.  As the $FB(d)$-model
 contains the undirected case, an analogous characterization of graph properties
for the $d$-bounded degree undirected graph model is implied.
We note that these are the first  results that do not restrict
the family of graphs, nor the family of testers under consideration
(apart of being
$1$-sided-error).

{\bf Related results:} 
There are many  results for the bounded-degree model on  the testability of specific
properties of graphs or digraphs, cf.  
\cite{BR02, CS10, GR00,  GR02, NS07, OR11, YI10, PR02}, 
and others. 
 In \cite{CPS16} the authors
relate (2-sided error) testability in the $FB(d)$ and $F(d)$ models. Other general results fall typically into three
categories. In the first not all $d$-bounded degree graphs are
considered, but rather a restricted family of graphs.  It is shown
e.g., in \cite{HKNO09, LR15, NS13} (and citations 
therein)\footnote{\cite{NS13} shows that any graph property is $2$-sided  error strongly-testable for any hyperfinite family of
graphs.}
that under certain restriction of the input graphs all graph properties are
$2$-sided error strongly-testable. 
The other two types of general results are when  the graph
properties under study are restricted, or the class of testers is restricted.
Most relevant for this work are the results of Czumaj, Shapira and
Sohler  
\cite{CSS09}, and Goldreich-Ron  \cite{GR09}.  
In \cite{CSS09} it is shown that any hereditary property is $1${\em-sided error}
strongly-testable if the input graph
belongs to a hereditary and non-expanding family of
graphs
 In
\cite{GR09} restricted $1$-sided error testers called
{\em proximity oblivious testers} (POT) for graph properties (and
other properties) are studied. The POT
is not being constructed for an explicitly given  distance parameter $\epsilon$. Instead, the tester
 works for any distance parameter
$\epsilon$, but its success probability deteriorates as 
  $\epsilon$ tends to $0$. \cite{GR09} give several general results to when graph
  properties have a POT in the bounded-degree model (and other models).

{\bf Techniques and description of results: }
Attempting for a characterization result  we
should  understand what are the limitations that
a 1-sided error test, making $O(1)$ queries, puts on the structure
of the property it tests. It turns out that this is relatively
simple. Using the tools from \cite{GR09} (see also
\cite{GT01}), one can transform any $1$-sided error
tester into a ``canonical'' one that picks (uniformly) $O(1)$
random vertices in $G$, and then scans the balls of radius $O(1)$
around each. Finally, it makes its decision based only on the subgraph $G'$
it discovers {\em and its interface to the rest of the graph}.  To
make this latter point clearer consider  the $3$-degree bounded model and  the
property of not having a  vertex of degree two.  This
property is $1$-sided error strongly-testable simply by looking at a random
vertex and rejecting if its degree is exactly $2$. Note
 that this decision  cannot be concluded just by the fact that
the subgraph seen is a {\em subgraph} of $G$. It is important that the
sampled vertex $v$
{\em is not connected} to any other vertex besides the $2$ 
discovered neighbours of it. Namely, this property is not specified by
a forbidden subgraph (or induced subgraph).  This
suggests the notion of {\em configuration} appearing also in
\cite{GR09}, and defined for our setting in  Section
\ref{sec:prelim}.  

Loosely speaking a configuration specifies an induced subgraph with
an induced ``interface''  to the rest of the graph (see Definition \ref{def:conf}). With this notion it is
fairly easy to see that any $1$-sided error test can essentially test only graph
properties that are close to being defined by a collection of forbidden
configurations (the additional subtleties arise from the fact that the
tester is actually being designed for a distance parameter $\epsilon$,
and for different $\epsilon$'s testers might reject different
configurations).

Is the converse true? Namely, is every property that is defined by a set of
forbidden configurations (let alone, being ''close'' to
such)  strongly-testable?  This is open at
this point. 

Showing that a property that is defined by a forbidden set of
configurations is $1$-sided error strongly-testable 
usually  amounts to proving what is  called ``removal
lemmas''. Namely, a lemma stating  that
if a graph is $\epsilon$-far from a property  then it has a {\em
  large} number of appearances  of forbidden configurations
 (here ``large'' is $f(\epsilon) \cdot n$, namely linear in $n$). 

In the case of monotone properties the notion of a `forbidden
configurations' can be replaced with `forbidden subgraphs'. A removal lemma is  true for monotone properties in all models. 
For hereditary properties `forbidden configurations' can be
replaced with `forbidden induced subgraphs'. A removal lemma is also
true for hereditary properties for the $FB(d)$-model, and in a slightly
different form for the $F(d)$-model, but is more complicated to prove. We
  use a somewhat
different argument and test for hereditary properties in the later case. 

Our main results show that for all the
bounded-degree models, for both
monotone properties, and hereditary properties, a property is
$1$-sided error strongly-testable  
if and only if it is ``close'' to a property that is defined by an
appropriate set of forbidden graphs (see Section \ref{sec:prelim} for
the exact definition of ``close'' in this context). It could be that
by replacing forbidden graphs with forbidden configurations, 
this becomes true for {\em any graph property}. If indeed true, this will settle the
characterization problem of $1$-sided error strongly-testable
properties (see the discussion at the end of Section \ref{sec:F}).
  We do not currently know if a generalization of some sort
is true even for undirected $3$-degree bounded graphs.

Finally,  the characterization that we present is a structural result
on $1$-sided error strongly-testable properties. It provides a 
better understanding of the different models and the difference
between them. 
One could further 
ask whether the characterization could be used to easily determine
whether  
a given property is $1$-sided error strongly-testable using arguments totally outside the area of property
testing. This is indeed demonstrated (Section \ref{sec:55}) by proving (the known results) that  2-colorability is not $1$-sided error
strongly-testable, and that not having a $k$-star as a minor is
strongly-testable (here $k$ is constant).

{\bf Organization: } We start with the essential notations and
preliminaries in Section \ref{sec:prelim}.  Section \ref{sec:results}
contains a statement of our main results for the $F(d)$-model, and
Section \ref{sec:F} contains the proofs of the main results.  Section
\ref{sec:5} contains further discussion, and examples of properties
that are strongly-testable but not monotone, neither
hereditary. Section \ref{sec:FB} contains the analogous
characterizations  for
the $FB(d)$-model. Finally Sections \ref{sec:55} and \ref{sec:concl}
contain the application of our results to simply prove some known
results, and some concluding remarks, respectively.


\section{Preliminaries}\label{sec:prelim}

\subsection{Graph related notations}
Graphs here are mostly directed, can have anti-parallel edges
but no multiple edges. We will describe the results (and corresponding
definitions) mainly for the
$F(d)$-model which is the more interesting technically. Moreover, as we
do not have a bound on the in-degree for this model, better understanding
 this model may form  a tiny step towards better understanding testing in sparse graphs (of
unbounded degree).

For a directed graph $G = (V,E)$  we denote by $(u,v)$ the {\em
  directed} edge $(u \rightarrow v)$. That is, $(u,v)$ is a forward edge from
$u$. In turn, $v$ will be a member in the outgoing list of neighbours of $u$.

\begin{definition} [Neighbourhood]
	For a digraph $G=(V,E)$ and a vertex $v \in V$ we denote by
        $\Gamma^+(v)$ the set of outgoing neighbours of 
        $v$. Formally, $\Gamma^+(v) = \{u~\mid ~(v,u) \in E\}$.

        Similarly, 
	$\Gamma^-(v) = \{u~\mid ~(u,v) \in E\}$ and $\Gamma(v) = \Gamma^+(v) \cup \Gamma^-(v)$. 
	\end{definition}
Note that for undirected graphs $\Gamma^+, \Gamma^-$ and $\Gamma$
coincide.

	We generalize the notion of neighbourhood for sets of vertices: For $S \subseteq V$ we denote by $\Gamma^+(S) = \{y \notin S~\mid ~~ \exists x \in S, ~(x,y) \in E \}$. $\Gamma^-(S)$ and $\Gamma(S)$ are defined analogously.

\begin{definition} [Degree Bound]
	For an integer $d,$ a digraph $G$ is called
        $d$-bounded-out-degree if for every $v \in V(G), ~
        |\Gamma^+(v)| \leq d$.
 The $F(d)$-model contains all digraphs
        that are $d$-bounded-out-degree.  
\end{definition}
Note that the in-degree of a vertex can be arbitrary.

For a (di)graph $G=(V,E)$ and $V' \subseteq V$, we denote by
$G\setminus V'$ the (di)graph on $V \setminus V'$ that is obtained form $G$ by deleting
the vertices in $V'$. We denote by $G[V']$ the induced subgraph of $G$ on
$V'$ (that is, $G[V']$ contains all edges in $E(G)$ with both
endpoints in $V'$). 

A directed $k$-star is the graph containing $k+1$ vertices $\{u_i, ~ i=0,
\ldots ,k\}$ and  the edges $\{(u_0,u_i)| ~ i=1, \ldots ,k \}$. In
this case $u_0$
is called the ``center''.

\subsection{Properties and testers}
\begin{definition}\label{def:f-model} {\bf (The $F(d)$-model; queries)}
	Let $G=(V,E)$ be a graph on $n$ vertices in the
        $F(d)$-model. The access to $G$ is via the following  oracle:
        A query specifies a name of a
        vertex $v \in [n]$. As a result the oracle provides
        $\Gamma^+(v)$ as an answer. 
\end{definition}
\vspace{-0.2cm}
Note that an algorithm has no direct access to the incoming edges of a specified vertex $v$.

\vspace{0.3cm}
We note that a standard query in the incident list model is 
for a pair $(v,i)$, where $v \in V(G)$ and $i$ an index, on which the
oracle's answer is the $i$th vertex in the ordered list $\Gamma^+(v)$. For $d=
O(1)$ the two query-types are asymptotically equivalent (up to multiplying the number
of queries by a factor of
$d$). We use the definition above to
emphasise that algorithms, as well as properties, are invariant to the order of the vertices in
$\Gamma^+$. 

The $FB(d)$-model is similar where for a query $v \in V(G)$, the
answer is the pair of sets $\Gamma^+(v)$ and $\Gamma^-(v)$ (both sets
are of size at most $d$). In the undirected case the result is
$\Gamma(v)$ (of size bounded by $d$).  In terms of property testing,
the $d$-bounded degree model for undirected graphs can be seen as a
submodel of $FB(d)$-model where each undirected edge is represented as
two anti-parallel edges.  

\ignore{
We note that while we consider that the input for a
tester is a {\em labeled} graph  in the corresponding model,
it is rather a data-structure for the graph. Namely, in which there
is an extra arbitrary order on $\Gamma^+(v), \Gamma^-(v)$ for every
vertex $v$. Hence, and unlike in the dense graph model, the same
labeled graph has many possible different representations, and hence inputs. When the tester
makes its query, it gets the neighbourhood order information, and it
might base its decision on it. This of course, seems useless, and will
not be important in the testers we design, however, for a
characterization result, as we need to consider all possible testers,
we need to take the above information into account.
}

\begin{definition} [(di)Graph Properties]
	A (di)graph property $P$ is a set of (di)graphs that is closed under isomorphism. Namely if $G \in P$ then any isomorphic copy of $G$ is in $P$. 
	 We write $P= \cup_{n	\in \mathbb{N}}P_n$, where $P_n$ is the set of $n$-vertex graphs in $P$.
\end{definition}

\begin{definition}[Graph distance, distance to a
  property, distance between properties]\label{def:dist}
  Let $G$ and $G'$ be (di)graphs on $n$ vertices in any of the
  $d$-bounded degree models (that is, the $F(d)$-model, $FB(d)$, or
  $d$-bounded degree undirected graph model). The distance,
  $dist(G, G')$, is the number of edges that needs to be deleted and /
  or inserted from $G$ in order to make it $G'$.
	
We say that $G, G'$ are $\epsilon$-far (or $G$ is $\epsilon$-far from $G'$) if
	$dist(G,G')$ $> \epsilon dn$. 
	Otherwise $G,G'$ are said to be $\epsilon$-close.
	
	Let $P_n, Q_n$ be properties of $n$-vertex (di)graphs. $G$ is
        $\epsilon$-close to $P_n$ if it is $\epsilon$-close to some
        $G' \in P_n$.  
        We say that $P_n$ $\rm{and}$  $Q_n$ are $\epsilon$-close (or
        $P_n$ is $\epsilon$-close to $Q_n$) if
    every graph in $P_n$ is $\epsilon$-close to  $Q_n,$ and every
 graph in $Q_n$ is $\epsilon$-close to $P_n$.
        %
\end{definition}

 \begin{definition}[Monotone properties and hereditary properties]\label{def:mon-hered-prop}
   A (di)graph property
  $P$ is monotone (decreasing) if for every $G=(V,E) \in P$, deleting
  any edge $e \in E(G)$ results in a (di)graph $G\setminus \{e\}$ that
  is in $P$.
	A (di)graph property $P$ is hereditary if for every $G=(V,E) \in P$ and  $v \in V(G)$,  $G \setminus \{v\} \in P$. 
\end{definition}

Many natural (di)graph properties are monotone, e.g., being acyclic,
being $3$-colourable etc.  Note that if $P = \cup_{n \in \N} P_n$ is a monotone graph
property then  for every $n \in \N,~$ $P_n$ is by itself monotone.

\begin{definition}[The (di)Graph Properties $\mathcal{P_{\mathcal{H}}}$
  and $P^*_{\mathcal{H}}$]\label{def:pH}
	Let  $\mathcal{H}$ be a  set of
        digraphs.
        A digraph $G$ is $\mathcal{H}$-free if  for every $H \in
       \mathcal{H}, $ $G$ does not contain any
       subgraph that is isomorphic to $H$.
       
        The monotone property $P_{\mathcal{H}}$
        contains all digraphs that are $\mathcal{H}$-free, and
        $P_{\mathcal{H}_n}$ contains all $n$-vertex (di)graphs in
        $P_{\mathcal{H}}$.
  Similarly, we denote by $P^*_{\mathcal{H}}$ the $\rm{hereditary}$
  property that is defined by being $\mathcal{H}$-free as $\rm{induced
    ~ subgraphs}$ and $P^*_{\mathcal{H}_n}$ the set of $n$-vertex
  (di)graphs in $P^*_{\mathcal{H}}$.

      \end{definition}

      \begin{definition}
        [bounded-size collections]\label{def:size}
        Let $\mathcal{H}$ be a set of (di)graphs. We call
        $\mathcal{H}$ a $r$-set if every member $H \in \mathcal{H}$ has
        at most $r$ vertices. 
      \end{definition}

\begin{remark}\label{rem:34}
	$ $
	\begin{itemize}
 \item A natural example of monotone decreasing
graph property is a property $P_{\mathcal{H}}$  that is defined by a  family of forbidden 
subgraphs $\mathcal{H}$. It is immediate from the definition that every monotone
graph property is defined by a family of forbidden subgraphs but this
family may be infinite.

Recall that   $P = \cup_{n \in \N} P_n$ is monotone if and
only if $P_n$ is monotone for every $n$. Namely, being monotone is
defined for every $n$ separately. In this respect  being monotone is
not a `global' feature of $P$ but rather a feature 
of the individual $P_n, ~ n \in \N$. In what follows it will important
to us how the individual monotone properties $P_n, ~ n \in \N$ are
defined. Obviously for any fixed $n$,  $P_n$ is defined by an
$r$-set of forbidden subgraphs, but $r$ may depend on $n$.  

To make this clearer, consider the property of being acyclic. This
property is defined by forbidding all di-cycles, which  is an infinite
family. For the individual slices $P_n, ~n \in \N,$ the
corresponding family although finite, it is {\em not} a $r$-set 
unless $r \geq n$.  An example of slightly different nature is that of
the monotone property that contains the digraphs that are not
Hamiltonian. For every $n \in \N,$ $P_n$ is defined by one forbidden
subgraph (the simple directed $n$-cycle). Thus $P_n$ is defined by a $n$-set of forbidden subgraphs but for no fixed $r$,
$P_n$ can be  defined by an $r$-set for every $n$.

This
distinction will become important in our characterization results.  It
will turn out that the strongly-testable monotone properties are
tightly related 
 to properties that are defined by $r$-sets of forbidden
subgraphs for $r$ that is  independent of $n$.

\item For family $\mathcal{H}$ of forbidden digraphs, the monotone property of
  being $\mathcal{H}$-free is determined by the minimal members of
  $\mathcal{H}$ (w.r.t edge deletions). That is, if for $H,H' \in
  \mathcal{H}$ it holds that  $H$ is a subgraph of $H'$,  then being
  $\mathcal{H}$-free is identical to being $(\mathcal{H} \setminus
  \{H'\})$-free. 
  
\item Hereditary (di)graph properties are very natural in graph theory. It
is immediate from the definition that a property is hereditary if and
only if it is defined by a 
collection (possibly infinite) of forbidden {\em induced} subgraphs. E.g.,
the property of not containing an induced (di)cycle of length
$4$, and the property of being bipartite (that is expressed in this
case as not containing an odd size cycle). Both these properties are
monotone and
hereditary. 

Hereditary properties are not necessarily monotone, and monotone
properties are not necessarily hereditary. Further, the feature of
being hereditary, unlike being monotone,  depends on the entire
property $P =\cup_{n \in \N} P_n$ and cannot be defined for a single
$n$-slice $P_n$.
\end{itemize}
\end{remark}

\noindent
{\bf Testers: }
We define here $1$-sided error  testers for  digraph properties in the
$F(d)$-model.

\begin{definition} [$1$-sided error $\epsilon$-test for a digraph
  property $P$, $F(d)$-model]

	A $1$-sided error test for a digraph property $P$ is a randomized algorithm
	that gets two parameters, $n = |V(G)|$ and a distance
        parameter  $\epsilon > 0$.  It accesses its input graph
        via vertex queries (Definition \ref{def:f-model}), and
        satisfies the following two conditions.
        \begin{itemize}
        \item It accepts every $n$-vertex digraph in $F(d)$ that belongs to $P$ 	with probability $1$.
          \item It rejects every $n$-vertex digraph that is $\epsilon$-far from $P$ with probability at least $1/2$.
        \end{itemize}
	\end{definition}

	 The query complexity of the test is the maximum number of queries it
         makes for any input
         graph (in $P$ or not in $P$) and for every run. Hence the query
         complexity is a function of $n$ and $\epsilon$.

         \vspace{0.5cm}
\noindent
{\bf A note on the definition of testers:} 
 A test
for a graph property $P$ is formally  an infinite set of tests
$\{T(\epsilon,n)\}_{n \in \N, \epsilon \in (0,1)}$, where
$T(\epsilon,n)$ is a test for $P_n$ and distance parameter $\epsilon$.  Namely, we deal here with a
non-uniform model of computation. We often use the term
         $\epsilon$-test to emphasize that the test is designed for an error
         parameter $\epsilon$.  This  will be of special importance in this
         paper, as for different distance parameters, the test will
         behave differently.  We are
 interested, as usual, in the query complexity $q$ as a function of
 $\epsilon$ and $n$. Note further that since our models are
 parameterized by $d$, the query complexity (or even the fact whether
 a property is testable in the corresponding $d$-bounded degree model)
 may depend on $d$.  We may state the query complexity dependence on $d$ but  this
 is  of no particular importance in this paper.

\begin{definition}[strong-testability]
  Let $Q:(0,1) \mapsto \N$. If a property $P$ has an $\epsilon$-test  whose  query complexity
   on every $n$-vertex  graph is bounded by $Q(\epsilon)$, we say that $P$
        is $\epsilon$-strongly-testable. If $P$ is
        $\epsilon$-strongly-testable for every $\epsilon \in (0,1)$ we
        say that $P$ is {\em strongly-testable}.
\end{definition}

\subsection{Configurations - the $F(d)$-model}
The following definition of {\em configuration} is of major importance in
this paper. The motivation behind the definition
is that a configuration is what a tester discovers after making some
queries to the graph. It will turn out that the configuration that a
tester discovers contains
 {\em all} the information that is
used by the tester in order to form its decision.

\begin{definition} [Configuration, $F(d)$-model]\label{def:conf}
	A configuration is a pair $C=(H,L)$, where $H=(W,F)$ is
        a $d$-bounded-out-degree graph, and $L$ is a function $L:
        W \rightarrow \{\mathrm{developed}, \mathrm{frontier}\}$. The out-degree of
        every frontier vertex is $0$. 
\end{definition}

Consider a run of a tester on a graph $G$. The tester discovers {\em all}  (the at most $d$)
outgoing neighbours of every queried vertex.  At the end of the run,
after  making $q$ queries, the tester discovers a subgraph $H$ of $G$. $H$
contains the $q$ vertices that are queried; these  correspond
to the $developed$ vertices in the configuration it discovers. $H$ may also 
contain vertices that are neighbours of queried vertices but that were
not themselves queried. These vertices are the $frontier$ vertices.   A frontier vertex that is
discovered by the tester and was not queried may
have outgoing neighbours, but the
corresponding edges (the forward edges from the frontier vertex) will
not be discovered by the tester.  Consequently, the out-degree of a
frontier vertex 
in the discovered configuration is $0$.  In contrast, all forward edges of  a
developed vertex are discovered. 

We now make the above formal using the defintion below.

\begin{definition} [$C$-Free, $F(d)$-model]\label{def:C-free}
	Let $C=(H,L)$  be a configuration, where $H=(W,F)$ a digraph and $L: W \rightarrow
        \{developed, frontier\}$.   Let $G=(V,E)$ be a
        digraph in the $F(d)$-model.
        We say that $G$ has a $C$-$\rm{appearance}$ if there is an injective mapping $\phi : W \rightarrow
        V$ with the following two properties:
        \begin{itemize}
        \item $\forall v,u \in W$ and $L(v)=\rm{developed}$,
        $(v,u) \in F$ $~\rm{ if ~  and~  only~   if}~ 
	(\phi(v),\phi(u)) \in E$.
  \item For every developed $v,$ if $(\phi(v), x) \in E$
then $\exists	u\in W, \phi(u)=x$.
        \end{itemize}
	
	We say that $G$ is $C$-free if $G$ has no $C$-appearance. 
\end{definition}

The notion of configuration (using slightly different terms) appears also in \cite{CSS09,GR09}. 

Let $C=(H,L)$ a configuration with $D \subseteq V(H)$ being the
developed vertices. Definition \ref{def:C-free} implies that if $G$ has
a $C$-appearance on a vertex set $V' = \phi(V(H))$, with $\phi$ being
the mapping as in the definition, then $G[\phi(D)]$ is isomorphic to
$H[D]$. Namely
$H$ induces an isomorphic digraph on its developed vertices as
$G$ does on the vertices that are the images of the developed
set of vertices $D$.  Further, the 2nd requirement in
Definition \ref{def:C-free} asserts that for every $v \in D$, all
forward edges of $\phi(v)$ in $G$ are the `images of edges' in
$H$. It is not necessarily  that 
$G[V']$ is isomorphic as an induced subgraph to
$H$. This is since there might be an edge $(x,y) \in
G[V(H')]$ that is not in $H$. This can happen only if
$x$ is an image of a frontier vertex. 

 To exemplify Defintion \ref{def:C-free}  further, consider $C=(H,L)$, where $H$ is the directed
 $2$-star and the center is the only developed
 vertex in $H$. A digraph $G$ has a $C$-appearance if and only if it has a vertex $v'$ with
 {\em exactly} two outgoing  neighbours $u_1',u_2'$.  There could be
 an edge $(u'_1,u_2') \in G$ and hence the subgraph that $G$ induces on
 $\{v',u_1',u_2\}$ might not be isomorphic to $H$. There could also be an
 edge $(x',v') \in E(G)$. However, there cannot
 be an edge $(v',y) \in E(G)$ where $y \notin \{u_1, u_2\}$.

 We sum up this discussion with the following obvious fact.

 \begin{fact}
   \label{fact:f3}
   Let $C=(H,L)$ be a configuration and $G$ a digraph (all with
   respect to the $F(d)$-model). Then:
   \begin{itemize}
   \item If $G$ has a  $C=(H,L)$-appearance then $G$ contains $H$ as a subgraph.
   \item    If $G[V']$ is isomorphic to $H$ as an
   induced subgraph,  then a subgraph of $G$ that is obtained by
   deleting $\Gamma^+(V')$ in $G$ has a $C$-appearance (for the
   given $L$). 
   \end{itemize}
 \end{fact}

Finally, looking towards a characterization theorem, it would be of use if we
could restrict the behaviour of possible testers to ``canonical'' ones. This proved useful
in the dense graph model in \cite{GT01} and it is of similar flavour (and
simpler) here. It was already done in \cite{GR09} for undirected
$d$-bounded degree graphs and the extension to directed graphs (in
both models) is straightforward. We state it here in order to
be consistent with our notations.

\begin{definition}[$r$-disc around a vertex, $F(d)$-model]
  Let $G$ be a digraph and $r \in \N$. The $r$-disc around $v \in
  V(G),$ denoted $~D(v,r)$,  is the subgraph of $G$ that is induced by all vertices $u$
  for which there is a path from $v$ to $u$  of length at most
  $r$.
  \end{definition}
We note that a tester can discover the $r$-disc  around a given
vertex $v \in V(G)$. This is  done by  making a
`BFS-like' search from $v$, where at each step the tester queries the
next first discovered but  not yet queried vertex  that is of distance less than
$r$ from $v$. Discovering $D(v,r)$ takes at most $d^{r}$ queries for a
graph in the $F(d)$-model. It is useful to consider such a
procedure as an augmented query, motivating the following
definition.

\begin{definition} [$r$-disc query, $F(d)$-model]
	An $r$-disc query is made by specifying a vertex $v \in V(G)$
        for which the answer is the  $r$-disc
        around $v$.
\end{definition}

\begin{definition} [canonical-testers]\label{def:canonical}
  A $(r,q)$-canonical tester for a graph property $P$ is a
  tester that chooses $q$ vertices uniformly at random
  $\{v_1, \ldots, v_q \}$. It then makes an $r$-disc query around 
  $v_i,$ for $ i=1, \ldots, q$. Then, depending only on the configuration
  it sees and possibly on $n$ (but not the order of the queries, or
  the internal coins) it makes its decision.
\end{definition}

The following result  \cite{GR09}, shows that
strongly-testable properties can  be tested by canonical-testers\footnote{In \cite{GR09} it is done only for undirected graphs,
  but the generalization to directed graphs in both models is straightforward.}.

\begin{thm} \label{thm:canonical}
  Let $T$ be a
  $1$-sided error 
  $\epsilon$-test for a digraph property $P$ in the
  $F(d)$-model.  If the query complexity of $T$ is bounded by $q$ then there is 
  a  $(q,q)$-canonical tester 
  that is a $1$-sided error $\epsilon$-test for $P$.
\end{thm}

Note that  a $(r,q)$-canonical-tester is  a {\em
  `non-adaptive'} algorithm with respect to $r$-disc queries.  

\ignore{
\begin{proof}[Sketch]
The proof follows the footsteps of \cite{GT01} where it is shown that
if $P$ has a 1-sided error tester of query complexity $q$, then $P$
has a canonical-tester that has 1-sided error and query complexity
${2q \choose 2}$.

	The test $T$ makes at most $q$ queries adaptively, and then based on
	the answers it gets, and in addition depending possibly on the names of the vertices,
	the order of the queries, the internal coins, and $n$- the graph
	size, it makes its decision. 
When $T$ makes a query to a vertex $v$, it might be that $v$ is
already discovered by previous answers (or queries), in which case we
call $v$ old. Or it might be
that $v$ is a new vertex that did not appear in queries or answers of
previous queries. We call $v$ new in this case.

We first simulate $T$ by an {\em adaptive} tester $T_1$ that acts as
follows: it first picks {\em adaptively} a set of $q$ vertices, and
perform a $q$-disc query on each. Obviously, such $T_1$ can be
constructed from $T$. Each time that $T$ picks a new vertex, $T_1$
picks the same new vertex (with the same probability), and performs a
$q$-disc query around it. Obviously, queries to old vertices can be
simulated having the information of the disc-queries. Finally, the
decision $T_1$ will take is the same that $T$ takes on the specific
run that is simulated.

We now simulate $T_1$ by a test $T_2$ that makes its $q$ disc-queries
to uniformly chosen random vertices. To do this we chose a 
 permutation on $n$ elements $\pi \in_R
\mathcal{S}_n$ uniformly at random. Then on an input graph $G$ on $n$
vertices, it behaves just like $T_1$, but
on $\pi(G)$. Namely, if $T_1$ makes a disc-query to a new vertex $v$ at a certain stage $T_2$ will
query $\pi(v)$. Finally, the decision that
$T_2$ will take on the resulting run is the same as $T_1$ would on the
information that is discovered.  Namely,  if one thinks of
$T_1$ as a distribution over deterministic trees, then each tree (with
its corresponding probability) induces $n!$ trees by choosing the extra
randomness, and behaving as described above.

It is easy to
see that $T_2$ is still a test for $P$, and further, if $T_1$
    is a $1$-sided error test, then so is $T_2$. The reason is that
    the probability that $T_2$ accepts $G$ is just the average
(according to the uniform distribution on $\pi \in S_n$), of the
acceptance probability of $G_{\pi}$, the permuted graph according to
$\pi$ (as this is true for every branch of $G_1$). However, since the property $P$ is a digraph
property, it is invariant under permuting as above. Hence, $T_2$
will accept any $G \in P$ with probability that is at
least the minimum acceptance probability of $T$ for any $G' \in
 P$, and reject $G$ that is $\epsilon$-far from $P$ with the minimum rejection probability
of $T_1$ for  any $G'$ that is $\epsilon$-far from
    $P$. 

We further note that the distribution induced on the vertices that are
chosen by $T_2$ to be queried is the uniform over all $q$-size subsets
of vertices. This can be formally proven by induction on the order of
vertices that are queried, but is also immediate from the fact that at
any step, the probability that a new $v$ is queried, is identical on
all yet non queried vertices. 

Hence, $T_2$ already makes its queries as a canonical-tester does. In
particular it is {\em non adaptive}. We
now get rid of the possible dependence of the decision $T_2$ reaches,
on the labels, order, or internal coins.  This is argued exactly as
in \cite{GT01}, since the symmetrization implies that
decisions is invariant to order or labels (removing dependence
on coin flips is trivial for $1$-sided error test; $T_2$ accepts on
seeing a configuration only if it accepts w.p $1$ on seeing this
configuration. For two sided error, a slightly more involved argument
is needed, and we should assume that the error probability of $T_2$ is
less than $1/6$ for $2$-sided error, see details in \cite{GT01}).
\end{proof}
}

\section{Our main results}\label{sec:results}
 We consider in what
follows the $F(d)$ model (for constant $d$). The $F(d)$-model is the more natural model from the algorithmic point
of view,  being consistent with the standard data structures for
directed graphs. It contains a strictly larger set of graphs than the
$FB(d)$-model (as the in-degree is not
bounded). 
From the property testing perspective it is more restricted
algorithmically due to the limited
access to the graph. 

We prove here that the strongly-testable {\em monotone} graph
properties are these that are
close (in the sense of Definition \ref{def:dist}) to be expressed by
an $r$-set of forbidden
subgraphs that have some additional connectivity
requirements. For hereditary properties the results are essentially the
same where forbidden subgraphs are replaced
with forbidden induced subgraphs. 
We need the following definitions.

\begin{definition} [Component]\label{def:component}
	Let $H=(V,E)$ be  a directed graph. A subset $V' \subset V$
        defines a {\rm component} of $H$,  if by disregarding the
        directions of the edges of $H$, $V'$ induces a 
	connected component in the resulting undirected graph.
        We say in this case that $H[V']$, the directed subgraph of $H$
        that is induced by $V'$, is a component of $H$. 
\end{definition}

We note that Definition \ref{def:component} is not a standard
graph-theory term, and we warn the reader not to confuse it with strongly
connected components of the digraph.  We are concerned with graphs of
multiple components as the forbidden graphs that define a monotone
property might be such. E.g., let $C_k$ be the directed $k$-cycle, and
consider the property $P_1$ of being $C_3$-free, $P_2$ the property of
being $C_4$-free, $P_3$ the property of being $\{C_3,C_4\}$-free, and
$P_4$ the property of being free of the single graph $H$ that is a
vertex disjoint union of $C_3$ and $C_4$. Namely, a graph is not in
$P_4$ if it has a $C_3$ subgraph and a disjoint 
$C_4$ subgraph. All properties $P_i, i=1,2,3,4$ are distinct. The properties
$P_1,P_2,P_4$ are defined by one forbidden graph. $P_3$ is defined by
two forbidden graphs. The forbidden graphs defining $P_1,P_2,P_3$ have
one component each, while the single forbidden graph defining $P_4$
has two components.

\begin{definition} [Rooted digraph]
	A digraph $H$ is {\rm rooted} if every 
        component $H'$ of $H$ has a  vertex $v$ such that
        for every $u \in V(H')$,  there is  a di-path from $v$ to $u$ in $H'$.
\end{definition}
We note that a digraph can have many roots. In particular, if it is
strongly connected then every vertex of it is a root. 
The significance of $v$ being a root in a component of size at most
$r$ is that 
making an $r$-disc query around $v$ will discover the whole component that
contains $v$.

\vspace{0.5cm}
Our main theorem, characterizing the strongly-testable monotone
properties is the following.
\begin{thm}
  \label{thm:main-mon}
  Let $P = \cup_{n \in \N}P_n$ be a monotone digraph property in the
  $F(d)$-model. Then $P$ is strongly-testable $\rm{if~ and~ only~ if}$ there is
  a function  $r: (0,1) \mapsto \N$ such that for any  $\epsilon >0$
  and $n \in \N,~$ 
  there is a  
 $r(\epsilon)$-set of rooted
digraphs $\mathcal{H}_n $  such that the property
$P_{\mathcal{H}_n}$ that consists of the 
$n$-vertex  digraphs that are $\mathcal{H}_n$-free, satisfies the
following two conditions:

\noindent
(a) $P_n \subseteq P_{\mathcal{H}_n} $\\
(b) $P_{\mathcal{H}_n} $ is $\epsilon/2$-close to $P_n$.
\end{thm}

We note that the sets  $\{\mathcal{H}_n \}_{n \in \N}$ in Theorem
\ref{thm:main-mon} may depend  on $\epsilon$  (as the
bound $r(\epsilon)$ depends on
$\epsilon$).

A Similar theorem for hereditary properties is the following.

Let $\mathcal{H}$ be  a set of digraphs. Recall the definition of the
property 
$P^*_{\mathcal{H}}$ from Definition \ref{def:mon-hered-prop}. We
denote by  $P^*_{\mathcal{H}_n}$ the set of $n$-vertex  digraphs in $P^*_{\mathcal{H}}$.

\begin{thm}
  \label{thm:main-hered}
  Let $P$ be an hereditary  digraph property in the
  $F(d)$-model. Then $P$ is strongly-testable $\rm{if~ and~ only~ if}$
  there are functions
  $r: (0,1) \mapsto \N$ and $N:(0,1) \mapsto \N$ such that for any  $\epsilon >0$
  there is a  $r(\epsilon)$-set of rooted
digraphs $\mathcal{H}$  such that for every $n \geq N(\epsilon),$
$P^*_{\mathcal{H}_n}$  satisfies the
following two conditions:

\noindent
(a) $P_n \subseteq P^*_{\mathcal{H}_n} $\\
(b) $P^*_{\mathcal{H}_n} $ is $\epsilon/2$-close to $P_n$.
\end{thm}

\vspace{0.4cm}
\noindent
{\bf Some comments on the results:}
    \begin{itemize}
    \item
      The lower bound  $n \geq N(\epsilon)$ in Theorem
\ref{thm:main-hered} is essential and not an artifact of the proof. 
 Consider  the $F(1)$-model and let $C_k$ be the directed cycle of
 size $k$. Let $P$ be the property that contains an $n$-vertex graph
 if it is free of all cycles $C_k$ for $k \leq \sqrt{n}$ (as induced
 subgraphs).   This is a strongly-testable hereditary (and monotone) property as
 asserted  by Theorem \ref{thm:main-hered} and the
 set $\mathcal{H}$ that contains all cycles up to size
 $\frac{1}{2\epsilon}$, for $N(\epsilon) = 4/\epsilon^2$.

 However, for any possible $r$-set $\mathcal{H'}$ for which $P
 \subseteq P^*_{\mathcal{H'}}$, for  $P$ to $\epsilon$-close to $P^*_{\mathcal{H'}}$, $\mathcal{H'}$ should contain all cycles of
 size at most 
 $1/\epsilon$. But then $P_n \subseteq P^*_{\mathcal{H'}}$
   only for  $n \geq 1/\epsilon^2$. 

 \item The `only if' direction of Theorem \ref{thm:main-mon} is
   restated as Theorem \ref{thm:main-ff2}. In Theorem
   \ref{thm:general} we generalize Theorem \ref{thm:main-ff2} by replacing
   the forbidden set of digraphs $\mathcal{H}_n$ with a finite set of
   forbidden configurations (see Definitions \ref{def:conf} and
   \ref{def:C-free} ). In turn, this stronger (and more immediate
   theorem) is true for any strongly-testable digraph property (rather
   than just for monotone).  Thus, Theorem \ref{thm:general} gives a
   necessary condition for {\em any} graph property to be $1$-sided
   error strongly-testable. For all we know, this could also be a
   sufficient condition. This will be further discussed in Section
   \ref{sec:concl}.
\item
One may ask whether
the extra restriction that $P_{\mathcal{H}_n}$ (or $P^*_{\mathcal{H}_n}$ in
case of hereditary property) is $\epsilon/2$-close to
$P_n$ rather than just being $P_n$ is a necessity or rather just an
artifact of our proof. The answer is that this is needed. 
Indeed, as mentioned in
the introduction, acyclicity is not  strongly-testable in
the $F(d)$-model for large enough $d$, even by $2$-sided error testes 
\cite{BR02}. 
  However, 
it is easy to see that directed acyclicity is $1$-sided error strongly-testable in the $F(1)$-model. 
\ignore{ Indeed a $1$-bounded degree digraph is composed of a
  collection of vertex disjoint paths and simple cycles. Hence if $G$
  is $\epsilon$-far from acyclic, it must contain at least
  $\epsilon n$ disjoint cycles. It follows that there are at least
  $\epsilon n/2$ cycles of length at most $2/\epsilon$. Hence,
  sampling a vertex at random and scanning the $2/\epsilon$ disc
  around it will discover such cycles with high probability. This
  immediately implies a $1$-sided error test of query complexity
  $O(1/\epsilon^2)$.  }
Acyclicity, while monotone, can not be defined by
an $r$-set of forbidden subgraphs in the $F(1)$-model for any fixed
$r$. Rather, it is
$\epsilon$-close (in the $F(1)$-model) to be $\mathcal{H}$-free as induced graphs for the
$\frac{1}{\epsilon}$-set $\mathcal{H}$ that contains all cycles of
size at most $1/\epsilon$.
\end{itemize}

\section{Proofs of the main results}\label{sec:F}
Here we prove Theorem \ref{thm:main-mon} and Theorem \ref{thm:main-hered}.
We will start by proving the `if' directions for both
theorems in Section \ref{sec:if}. Section \ref{sec:mon-inverse}  contains the proofs of the
`only-if' parts.

\subsection{Monotone properties and hereditary
  properties that are strongly-testable}\label{sec:if}

Theorem \ref{thm:main-mon} states that if $P=\cup_n P_n$ is $\epsilon$-close
to $P_{\mathcal{H}_n}$ for an $r$-set of rooted digraphs $\mathcal{H}_n$
  then $P$ is $1$-sided error strongly-testable. We start by proving that
 the monotone property
 $P_{\mathcal{H}}$ itself is strongly-testable  for a fixed $r$-set $\mathcal{H}$.

 Let $\mathcal{H}$ be a $r$-set of digraphs and
 $P$ the monotone property that contains the 
digraphs that are $\mathcal{H}$-free. Remark \ref{rem:34} implies that we may assume in what follows that
$\mathcal{H}$ does not contain two graphs such that one is a subgraph of
the other. We also note that if $\mathcal{H}$ contains a graph that is an isolated vertex (or a set of
isolated vertices) then $P_\mathcal{H}$ becomes trivial (empty for
large enough $n$). We assume in what follows that the above does not happen.

We start with the following preliminary proposition for the subcase of
Theorem \ref{thm:main-mon},  where $P = P_{\mathcal{H}}$.
\begin{proposition} \label{theorem:directed-out:1} Let $\mathcal{H}$ be a
  $r$-set of rooted digraphs and $|\mathcal{H}|=t$. Then the monotone
  property $P = P_{\mathcal{H}}$ has a $1$-sided error $\epsilon$-test
  in the $F(d)$-model, making $O(t r^2 d^{r+1} \ln r/\epsilon)$
  neighbourhood queries.
\end{proposition}
\begin{proof}
	The top level idea is simple, and a similar idea was used in
        \cite{GR02}:  Suppose that  a digraph $G$  
         is $\epsilon$-far from being ${\mathcal{H}}$-free.
        We will show that there is a
        large set of vertices, each being a root in an
        $H$-appearance in $G$ for some $H \in \mathcal{H}$. Hence 
        sampling of a random vertex and scanning the $r$-disc around it will
        find a forbidden $H$-appearance in $G$. Some extra care should be taken for
        disconnected forbidden subgraphs.  

        Formally, we prove that the following test $T(\epsilon, n)$ is a test
	for $P_{\mathcal{H}}$.
	
	{\bf   $T(\epsilon, n)$: }
	Repeat for $\ell =  (tr^2 d/\epsilon) \cdot 2\ln r$
        times independently:  Chose  a vertex $v \in_R V(G)$ uniformly
        at random and
        make an $r$-disc query around $v$.
		If  some $H \in
	\mathcal{H}$ is found as a subgraph in the discovered subgraph
        of $G$ then reject.  Otherwise accept.

Obviously the test accepts with probability $1$ every graph that is
$\mathcal{H}$-free. Further, the claimed complexity is clear.

	Assume that $G$ is a digraph on $n$ vertices that is $\epsilon$-far from
	$P_{\mathcal{H}}$.  We claim that  $G$ contains at least $\epsilon n/r$ edge
	disjoint subgraphs, each that is isomorphic to some 
	$H \in \mathcal{H}$. This is so as let $F$ be any maximal edge
        disjoint collection of
        subgraphs of $G$, each that is isomorphic to some $H \in
        \mathcal{H}$.  By deleting all outgoing-edges
         that are adjacent to vertices in $F$ (at most $|F|\cdot
        r \cdot d$) none of the subgraphs in $F$ is a  forbidden
        subgraph anymore. Further,  no new forbidden
        subgraph is created (by the assumption that no graph in
        $\mathcal{H}$ is a subgraph of another graph in $\mathcal{H}$).
 Therefore,  $G$ becomes $\mathcal{H}$-free after
        deleting these edges. We conclude that 
        $|F|\cdot r\cdot d \geq \epsilon nd$.

 Fix such a collection of subgraphs $F$. We  deduce that there is some fixed
	graph $H \in \mathcal{H}$ that is isomorphic to at least
	$\frac{|F|}{t} \geq \epsilon n/(tr)$ of the  digraphs in $F$. Fix such
	$\epsilon n/(tr)$ edge disjoint subgraphs in $G$, which we refer to as $F'$.
	
	Assume first that $H$ is composed of one single rooted
        component. Since the subgraphs in $F'$ are edge
        disjoint, a root vertex $v$ can appear in at most $d$ such
        distinct subgraphs (on account that it must have at least one
        forward edge in each such appearance).  We conclude that  there are at least
        $\frac{|F'|}{d} \geq \frac{\epsilon n}{t r d}$ distinct
        vertices, each being a root in an $H$-appearances in $G$. Hence, with probability
        $\frac{\epsilon}{trd}$ a random vertex $v$ 
         will be one of these roots. Assuming that such a
        vertex $v$ is chosen by $T(\epsilon,n)$, then making the
        $r$-disc query to $v$ will discover 
        the corresponding  $H$-appearance. Thus the failure probability is bounded by
        $(1-\frac{\epsilon}{trd})^\ell < 1/2$.
	
	Finally, assume that $H$ is composed of several rooted
        components. Since $|H| < r$, $H$ is composed of at most $r$
        components $C_1, \ldots C_a,~ a \leq r$. In this case, finding
        $a$ vertices $v_1, \ldots v_a$, with the $i$th being the root
        of a subgraph isomorphic to $C_i$ will discover an isomorphic
        copy of $H$ in $G$.  The probability of sampling a root of a
        component of type $C_i$ is at least
        $\frac{\epsilon}{t r^2 d}$. The union-bound implies that the 
        probability that there exists some type that we don't sample a
        root of is at most
        $a\cdot (1-\frac{\epsilon}{t r^2 d})^{\ell} \le 1/2$. This
        concludes the proof.
	\end{proof}

   It is assumed  implicitly in
       Proposition \ref{theorem:directed-out:1}  that $\mathcal{H}$ is a
       collection of digraphs in the $F(d)$-model. Therefore, the fact that $\mathcal{H}$
       is an $r$-set  implies  that $t=|\mathcal{H}|$ is bounded in terms of $r$
       (exponentially). Although not of prime interest for this paper,
       we still give the above tighter dependence on $t$  because  $t$ could be much smaller than
       the worst case bound.

For hereditary properties a Proposition analogous to Proposition \ref{theorem:directed-out:1}   will be stated. In this
case being $\mathcal{H}$-free as subgraphs is replaced by being free
as {\em induced}
subgraphs. However, unlike the easier case of monotone properties, we
can't 
assume that if $G$ is $\epsilon$-far from the property, then it
contains many vertices that are roots of 
$\mathcal{H}$-appearances. The reason is that  deleting 
edges in an $\mathcal{H}$-appearance in $G$ may create a new
$\mathcal{H}$-appearance\footnote{It could be true that for every
  $\mathcal{H}$, if
  $G$ is far from being $\mathcal{H}$-free as induced subgraphs, then there are many
  $\mathcal{H}$-appearances in $G$, but we 
  do not have a proof nor a counter example for this.}. We
use a  different argument.

\begin{definition}
  Let $\mathcal{H} $ be a  set of
  digraphs. We say that $H \in \mathcal{H}$ is essential if the
  digraph $H$ is $(\mathcal{H} \setminus \{H\})$-free as induced
  subgraph. Namely, $H$ does not contain as an induced subgraph any member of
  $\mathcal{H}$ except for itself.  If every $H \in \mathcal{H}$ is
  essential, we say that $\mathcal{H}$ is non-redundant.
\end{definition}

\begin{proposition} \label{thm:main-ff0.5} Let
  $\mathcal{H}$ be a non-redundant $r$-set of
  rooted digraphs. 
  Then the 
  hereditary property of being $\mathcal{H}$-free as induced subgraphs
  is $1$-sided error strongly-testable in the $F(d)$-model.
\end{proposition}

The following lemma is folklore. We state it for completeness.

\begin{lemma}[sampling a random edge]
  \label{lem:sample}
Let $G = (V;E)$ be a graph in the $F(d)$-model with $|E(G)| \geq \epsilon nd$.
Then, with probability at least $\epsilon/d$, the following randomized algorithm 
outputs an edge $e \in E$ that is distributed uniformly in $E$, 
and outputs a special failure indication otherwise. 
The algorithm sample a vertex $v \in V (G)$ uniformly at random,
queries this vertex to obtain $\Gamma^+(v)$,
and outputs each edge going out of $v$ with probability $1/d$.
In other words, letting $k=|\Gamma^+(v)|$, the algorithm stops
indicating failure with probability $1-\frac{k}{d}$,
and otherwise it samples $u\in\Gamma^+(v)$ uniformly at random and outputs $e = (v,u)$.
\end{lemma}


\begin{proof}
  Since $|E(G)| \geq \epsilon nd$ there are at least $\epsilon n$ vertices
  each with outdegree at least $1$. Let this set be $V_1$. The algorithm will output an edge in
  the case it chooses $v \in V_1$, and that it does not choose to
  indicate failure after choosing $v$.  This occurs with probability at least $\epsilon /d$.
  
The algorithm outputs a fixed edge $e=(v,u)$ with probability
    $Pr(e) = \Pr(v) \cdot \frac{deg(v)}{d} \cdot \frac{1}{deg(v)}
    = \frac{1}{|V_1|d}$. Since this is identical for all edges, the
    algorithm induces the uniform distribution on $E(G)$. 
\end{proof}

\begin{proof}[of Proposition \ref{thm:main-ff0.5}]
  For this proof, we abbreviate ``$H$-appearance'' and
``$\mathcal{H}$-appearance'' for $H$-appearance as {\em induced}
subgraph, and $\mathcal{H}$-appearance as induced subgraphs,  respectively.

  We may assume that $\mathcal{H}$ does not include an isolated vertex
as a member, as otherwise, being $\mathcal{H}$-free is an empty
property.  Further, we may assume that for no  $H \in \mathcal {H},$
$H$ contains an isolated vertex. As otherwise, we replace such $H$
with $H'$ that is obtained from $H$ by removing the isolated
vertices. Obviously, for $n$ large enough,  $G$ contains $H$ as an
induced subgraph if and only if $G$ contains $H'$ as induced
subgraph.

The test samples some vertices and scans the $r$-disc around each. It
rejects only if it finds a $\mathcal{H}$-appearance in the subgraph of
$G$ that it discovers.  The
vertex set that is sampled is a set of endpoints of
$\ell= \frac{8td\ln r}{\epsilon }$ random edges. This is done by
calling the algorithm of Lemma \ref{lem:sample} for $4d\ell/\epsilon$
times. Note that the lemma guarantee a success probability of
$\epsilon/d$ per edge query only for graphs with
$|E(G)| \geq \epsilon d n$ edges. In general, these $4d\ell/\epsilon$
calls could result in some random edges or none at all. If less than
$\ell$ edges are produced by the $4d\ell/\epsilon$ calles to the
algorithm in Lemma \ref{lem:sample}, the algorithm will stop and
accept.  Thus the overal query complexity is
$O(d^2t \ln r/\epsilon^2)$ neighbourhood queries in addition to
$O(td \ln r/ \epsilon)$ $r$-disc queries.

It is clear that for $G$ that is $\mathcal{H}$-free the test accepts
with probability $1$. 

Let $G$ be a digraph on $n$ vertices that is $\epsilon$-far from
being ${\mathcal{H}}$-free as induced subgraphs. Since $G$ must be
$\epsilon$-far from the empty graph, it follows that $|E(G)| \geq \epsilon dn$. 
This implies that
with probability at least $7/8$ the  $4d\ell/\epsilon$ calls to the
algorithm in Lemma \ref{lem:sample} will indeed produce at least
$\ell$ random edges. In what follows we condition the analysis on the
assumption that indeed $\ell$ random edges are produced.

For simplicity we first analyze  the test for the case that each $H \in
\mathcal{H}$ has only one rooted component (i.e, this does not cover,
e.g., the property of being free of a disjoint pair of a di-triangle
and a $4$-cycle). The argument for the general case will be
 somewhat harder.


	Let $S$ be a maximal set of  subgraphs of $G$, each being 
        an $\mathcal{H}$-appearance, and  in which
        the {\em forward-edges
        of the roots are
        disjoint}.  For each subgraph in $S$ fix one root vertex. Let
      this set of vertices be $R$.

   Assume first that $|S| \ge
        \frac{\epsilon n}{2}$. Then for an edge $e=(u,v)$, sampled uniformly at
        random from $E(G)$,  $u$ is a root of an $\mathcal{H}$-appearance with probability at least
        $p_1 = \frac{\epsilon}{2d}$. Hence, choosing $\ell$ random
        edges will  find a vertex that is a root of an
        $\mathcal{H}$-appearance with 
        probability of at least $3/4$.

  Suppose now that $|S|\le \frac{\epsilon n}{2}$. Then  $|R| \le
        \frac{\epsilon n}{2}$ (as we fixed one root vertex per member
        in $S$). Let $E^-(R) = \{(u,v) \in G~ | ~ v \in
        R \}$.

Assume first that $|E^-(R)| < \frac{\epsilon nd}{2}$. Let $E(R)$
be the set of all edges adjacent to $R$ (both incoming and outgoing
edges). Then 
$|E(R)| \le
d |R| + |E^-(R)|  <
                  \epsilon nd$. Therefore deleting all edges
                  in $E(R)$  results in a subgraph in which the vertices in $R$ 
                become  isolated and 
                all old $\mathcal{H}$-appearances in $S$ will be
                destroyed. We claim that the resulting graph $G'$ becomes
                $\mathcal{H}$-free. Indeed if $G'[V']$ is isomorphic to
                some $H \in \mathcal{H}$, either $G[V']$  is also so,
                or it is created by the absence of some old edges that
                are deleted. In the first case, 
                $G[V']$  must share an edge $(u,v)$ with an
                appearance in $S$, and where $u$ is a root in both
                appearances. This cannot happen as the edge $(u,v)$ is
                deleted. For the second possibility, as we delete {\em
                  all} edges (forward and backwards edges) adjacent to
                roots, deleting an edge $(u,v)$ makes $u$ isolated
                in $G'$ and hence, by the discusion in the first
                paragraph of the proof, $u$ cannot be part of an
                $\mathcal{H}$-appearance.

                The fact that  $G'$  becomes
                $\mathcal{H}$-free  is in contradiction with
                the assumption that $G$ is $\epsilon$-far from being
                such, as we have deleted less than $\epsilon dn$ edges.
Hence $|E^-(R)| \geq  \frac{\epsilon nd}{2}$. 
	But then sampling a random edge $e \in E(G)$  will
              result in $e=(u,v)$ for which $ v \in R$
              with success probability at least $\epsilon/2$. Thus, choosing $\ell$ random edges
                implies that we pick a root of an $\mathcal{H}$-appearance with
                probability  at least $3/4$. 


                We conclude that in all cases (of sizes of $S$) we find a vertex that is
                a root vertex of an $\mathcal{H}$-appearance with
                probability at least $3/4$. If this happens,  then scanning the $r$-disc
                around the endpoints of the sampled edges  will
                discover the $\mathcal{H}$-appearance.  This concludes  the
                proof  for this simple case (in which each $H \in
                \mathcal{H}$ has a single rooted component).  

                \vspace{0.4cm}
                \noindent
{\bf The general case:} For the general case, the same argument does
not 
 work directly. To realize what is the difficulty, assume that a
forbidden graph $H$ consists of two components:  a di-triangle and a disjoint $4$-cycle.
Assume also that $G$ is $\epsilon$-far from being $H$-free and that
 there is a small
number of $H$-appearances in $G$.  Then, similarly to the second
case  above,  we conclude   that $E^-(R)$ is large, where $R$ is the
set of roots of the $H$-appearances. This would mean that we
can find a root vertex in an $H$-appearance by making only a small
number of queries. But what if most of these edges are going into vertices
in di-triangles, and only very few to vertices in $4$-cycles.  In order to discover a forbidden
subgraph we also need to discover a $4$-cycle. In the general case
we need to
combine more carefully the several cases of different sizes of $E^-(R)$. This we do as follows:

Let $G$ be  a digraph on $n$ vertices that is $\epsilon$-far from being
$\mathcal{H}$-free as induced subgraphs (where we no longer assume that each forbidden
graph in $\mathcal{H}$ has only one component).

For $\mathcal{H} = \{H_1, \ldots ,H_t\},$ let $H_i$ be composed of
disjoint components $H_{i,j}, j=1, \ldots j_i$. 
Let $S$ be a maximal set of  subgraphs of $G$, each being 
        an $H_{i,j}$-appearance for some $i,j$, and  in which
        the {\em forward-edges
        of the roots are
        disjoint}.  

      We can write $S = \cup_{i,j} S_{i,j}$ where $S_{i,j}$ contains
      the corresponding appearances  of $H_{i,j}$ in $G$. Let
      $R_{i,j}$ be the set of the 
      corresponding roots, one per each appearance in $S_{i,j}$, and $\gamma_{i,j} = |E^- (R_{i,j})|$. Note that $i$ ranges over $\{1,\ldots ,t\}$
      and $j$ ranges over all possible components types of $H_i$ which is a
      number $j_i, ~ j_i \in \{1, \ldots ,r\}$.

For each $i \in \{1, \ldots , t\}$ let $I_i = \{j \in  \{1, \ldots
      ,j_i\}~ | ~ |S_{i,j}| < \delta n=\frac{\epsilon n}{2t} \}
      $.
      
 {\bf case (a): }    
      Assume that for some $i \in \{1, \ldots ,t\}$, for every
      $j \in I_i$, $\gamma_{i,j} \geq \frac{\epsilon d n}{2t}$.

      In this case,  for every $j \notin I_i$, for a random edge
      $(u,v)\in E(G)$, $u$  is
      going to be a root of an $H_{i,j}$ appearance (namely in
      $R_{i,j}$) with probability at least $\delta /d = \frac{\epsilon}{2td}$.   In addition, for
      every $j \in I_i$, a random edge $(u,v)$ picked uniformly from
      $E(G)$ will have $v \in R_{i,j}$ with probability at least
      $\frac{\gamma_{i,j}}{d^2n} = \frac{\epsilon}{2td}$ (as $v$ could be  a root of at most $d$
      distinct members in $S$).

   Hence sampling $\ell > 4 \ln r \cdot \frac{2td}{\epsilon}$ random
   edges implies that a root in an appearance of
$H_{i,j},$ for every $j \in \{1, \ldots ,j_i\},$ will be found with
probability at least $7/8$.
Calling  the sampling algorithm of
Lemma \ref{lem:sample} for $4d \ell /\epsilon$ times results in at least $\ell$ random edges with
probability at least $7/8$. Therefore, the overall success probability in
this case is at least $3/4$. 

{\bf case (b): }
 If case (a) does not hold, then for every $i
\in \{1, \ldots , t\}$, there is $j(i)\in I_i$ for which 
$\gamma_{i,j(i)} < \frac{\epsilon d n}{2t}$. (It could be that for some $i$ there are more
than one $j(i)$ as above; in that case, choose an arbitrary one.) 

But then  deleting, for every $i
\in \{1, \ldots , t\}$,  all edges incident to every
root in $S_{i,j(i)}$ (forward and backward edges),  all
$\mathcal{H}$-occurrences in $S$ will be destroyed (as for each $H_i$
we have destroyed 
all appearances of $H_{i,j(i)}$ in $S$). Moreover, no new appearances are created by the same
reasoning as in the simple case. Finally, we have deleted at most
$\sum_{i=1}^t d|S_{i,j(i)}| + \gamma_{i,j(i)} < \epsilon dn$ edges which
contradicts the assumption that $G$ is $\epsilon$-far from being $\mathcal{H}$-free.    
\end{proof}

We have proved so far that monotone or hereditary properties that are
defined by an $r$-set of forbidden rooted digraphs are strongly-testable. To
prove the `if-part' of Theorems \ref{thm:main-mon} and \ref{thm:main-hered},  we will also show
that  properties that are {\em close}  to such properties are
strongly-testable.  This is done next.
The following is a restatement of the `if-part' of Theorem \ref{thm:main-mon}.
\begin{thmm}
  \label{thm:main-ff1}
Let $\mathcal{H}$ be a $r$-set of rooted
digraphs and for $n \in \N$ let $P_{\mathcal{H}_n} $ the monotone property that contains all
$n$-vertex digraphs that are $\mathcal{H}$-free as subgraphs. Let
$P=\cup_n P_n$ be a digraph 
property in the $F(d)$-model for which, (a) $P_n \subseteq P_{\mathcal{H}_n} $, and (b)  
$P_{\mathcal{H}_n} $ is $\epsilon/2$-close to $P_n$. Then, $P$ is $1$-sided
error $\epsilon$-strongly-testable in the $F(d)$-model.
\end{thmm}

\begin{proof}
  By Proposition \ref{theorem:directed-out:1}, for every $\delta >0$ there
  is a $1$-sided error $\delta$-test for $P_{\mathcal{H}}$. Let
  $\delta= \epsilon/2$ and $T$ be a corresponding $1$-sided error
  $\delta$-test for $P_{\mathcal{H}}$.  We run $T$ on $G$, accept if $T$ accepts and reject
  otherwise. If $G \in P_n$ then since $P_n \subseteq P_{\mathcal{H}}$
  the test will accept $G$ w.p. $1$.  On the other hand, if $G$ is
  $\epsilon$-far from $P_n$, then it must be $\epsilon/2$-far from
  $P_{\mathcal{H}}$ as $P_{\mathcal{H}_n}$ is $\epsilon/2$-close to
  $P_n$. Hence, $G$ is rejected with probability at least $1/2$.
\end{proof}

We state below the corresponding restatement of `if-part' of Theorem
\ref{thm:main-hered}. Its proof is identical to that of Theorem
\ref{thm:main-ff1},
where we replace
Proposition \ref{theorem:directed-out:1} with Proposition
\ref{thm:main-ff0.5}. 

\begin{thmm}
  \label{thm:main-ff1.5}
Let $\mathcal{H}$ be a non-redundant $r$-set of rooted
digraphs and for $n \in \N$ let $P^*_{\mathcal{H}_n} $ the hereditary property that contains all
$n$-vertex digraphs that are $\mathcal{H}$-free as induced subgraphs. Let
$P=\cup_n P_n$ be a digraph 
property in the $F(d)$-model for which  (a) $P_n \subseteq P^*_{\mathcal{H}_n} $ and (b)  
$P^*_{\mathcal{H}_n} $ is $\epsilon/2$-close to $P_n$. Then $P$ is $1$-sided
error $\epsilon$-strongly-testable in the $F(d)$-model. $\qed$
\end{thmm}


\begin{remarkp}\label{rem:r18}
  $ $
  \begin{itemize}
  \item 
  Theorem \ref{thm:main-ff1} is stated in
  terms of a fixed family of forbidden digraphs $\mathcal{H}$. However,
  since the conditions (a) and (b) in the theorem are in terms of
  the slices $P_n$, namely  for $n$-vertex graphs,  the
  family $\mathcal{H}= \{\mathcal{H}_n \}$ may depend on $n$. The only global requirement
  of $\mathcal{H}_n$ is that it is an $r$-set, where $r$ is a function
  of $\epsilon$ only.

To make this clearer consider e.g., the  
property $P$ in the $F(d)$-model that contains every $n$-vertex graph $G$ if
$n$ is even,  and contains the digraphs that 
do not have a directed $4$-cycle otherwise.  $P$ is monotone but
it is not defined
by a single set of forbidden subgraphs. Rather,  for
every $n,$ $P_n$ is a slice of a property that is defined in this way. Hence, $P$ is 
$1$-sided error strongly-testable.

\item Note that the digraph property $P$ that is asserted to be
  strongly-testable in Theorem \ref{thm:main-ff1} is not necessarily 
  monotone. It is only required that it is close to a monotone
  property. In this sense, Theorem \ref{thm:main-ff1} is slightly
  stronger than the `if-part' of Theorem \ref{thm:main-mon}. An
  analogous 
  remark also holds for the property $P$ in Theorem
  \ref{thm:main-ff1.5}

  \item Note that in the characterization theorem, Theorem
    \ref{thm:main-hered}, we did not restrict the family $\mathcal{H}$
    to be non-redundant. This is not need as it is clearly the case
    that $P_{\mathcal{H}}^* = P_{\mathcal{H}'}^*$ for $\mathcal{H}'$
    that is obtained from $\mathcal{H}$ by removing all non-essential graphs. 
\end{itemize}
\end{remarkp}

\subsection{The `only-if' parts of Theorems \ref{thm:main-mon} and \ref{thm:main-hered}}\label{sec:mon-inverse}

Theorem \ref{thm:main-mon} requires 
that the corresponding family $\mathcal{H}$ contains members that are rooted. 
  We first show why this 
restriction is needed.  We say that $H \in \mathcal{H}$ is {\em
  minimal} if there is not $H' \in \mathcal{H} \setminus \{H\}$ for
which $H'$ is a subgraph of $H$.

\begin{proposition} \label{tester:directed-out:2} Let
  $\mathcal{H} = \{H_1, \ldots ,H_t\}$ be a set of forbidden digraphs
  and $P_{\mathcal{H}}$ be the corresponding monotone property of $n$-vertex graphs.
 If for some minimal
  $H \in \mathcal{H}$, $H$ is not rooted, then any 1-sided error
  $\frac{1}{d|H|}$-test for $P_{\mathcal{H}}$ makes $\Omega(\sqrt{n})$
  queries in the $F(d)$-model.
\end{proposition}
\vspace{-0.3cm}

\begin{proof}
Assume that $H \in \mathcal{H}$ is minimal and not
        rooted.  Set $\epsilon = \frac{1}{d|H|}$.
	An $\epsilon$-test for $P_{\mathcal{H}}$ that is $1$-sided
        error must discover some $H \in \mathcal{H}$ on any run that
        rejects. Hence it is enough to prove that any test that
        discovers a $\mathcal{H}$-appearance and makes $o(\sqrt{n})$
        queries must have a success probability that is less than
        $1/2$ on some $n$-vertex graphs that are $\epsilon$-far from $P_{\mathcal{H}}$.
	
	We use Yao's principle to prove the lower bound. Namely, we construct a
	probability distribution $\mathcal{D}$ that is supported on
	 $n$-vertex digraphs in $F(d)$ that are
	$\epsilon$-far from $P_{\mathcal{H}}$.
	We then 
	show  that  any {\em deterministic} algorithm making
	$q <\sqrt{\frac{n}{3|H|}}$  queries fails to find a copy of $H \in
        \mathcal{H}$ for more than $1/2$ of the inputs weighted according to $\mathcal{D}$.
	
	Let
	$G= (V,E)$ be an unlabelled directed graph on $n$ vertices that is a union of
	$\frac{n}{|H|}$ vertex disjoint copies\footnote{If $|H|$ does not divide $n$, we  augment
	$G$ with at most $|H|-1$ isolated vertices to get an
        $n$-vertex graph.}
	of  $H$ .  The
	distribution $\mathcal{D}$  is formed by labelling $V$ according
	to a random permutation uniformly chosen from the set of all
	permutation on $n$ elements.  Obviously $\mathcal{D}$ is supported on
	$\epsilon$-far graphs. Moreover, the only forbidden subgraphs 
	in each graph supported by $\mathcal{D}$ are disjoint copies of
	$H$. Hence, any deterministic $1$-sided error test with
        respect to $\mathcal{D}$
	ends correctly only when it finds a copy of
	$H$.
	
		Let $A$ be any deterministic algorithm making $q$ queries,
	adaptively. Every query made by $A$ is  of the form $v \in [n]$, where $v$
	is either one
	of the vertices that occurred as answers for some prior
	queries,  or $v$ is a new vertex that was not yet seen.  
      We will augment the algorithm so that on query $v$, the algorithm receives the entire
	subgraph $H_v$ containing all vertices {\em reachable from} $v$ in the copy
	of $H$ where $v$ lies. Note that this gives more information to the
	algorithm in the form of possibly $|H|-2$ additional vertices
        but with at least one vertex $w$ in the $H$-appearance of $v$
        that is excluded  by the assumption that
        $H$ is not rooted.
 	Hence, if the augmented algorithm does not discover a copy
	of $H$ neither does $A$. Note further that the additional information makes
	the queries of the first type -- namely, queries to vertices that are
	the answers to prior queries redundant. 
	
	Hence the augmented algorithm will end correctly after making
        $q$ queries $v_1, \ldots v_q$ only if it for some distinct
        $i,j \in \{1, \ldots ,q\}$, the vertices $v_i$ and $v_j$ belong to the same
        component of $G$ but none is reachable from the other.  This
        probability is clearly bounded by
        ${q \choose 2} \cdot \frac{|H|}{n} < 1/2$, for our choice of
        $q$ and $n$ large enough.
	\end{proof}

\subsubsection{The `only-if' part of Theorem
\ref{thm:main-mon}}        
Proving the `only-if' part of Theorem
\ref{thm:main-mon}  naturally brings us back to
configurations in digraphs as this is what a tester discovers in its
run. This motivates the following definition analogous to Definition
\ref{def:pH}.



\begin{definitionp}\label{def:pHp}
	For a set of
        configurations $\mathcal{C}$,
        the property $\mathcal{P_{\mathcal{C}}}$
        contains all graphs that are $C$-free for  every  $C \in
        \mathcal{C}$.
      \end{definitionp}

      We comment that for an unrestricted set of forbidden configurations
      $\mathcal{C}$, $P_{\mathcal{C}}$ may happen to be hereditary,
      monotone, or neither (in the $FB(d)$-model, $F(d)$-model and the
      undirected bounded-degree graph model).  E.g.,  the property of
      not having a vertex of out-degree
      exactly $2$ in the $F(3)$-model is a property that is defined by one forbidden
      configuration that is the directed $2$-star, where the center is
      the only developed vertex.
      However, the property is not  monotone nor
      hereditary (and  happens to be strongly-testable).

The following is a restatement of the `only if' part of Theorem
\ref{thm:main-mon} followed by its proof. Note that configurations do
not appear in the statement, but will appear in the proof.

\begin{thmp} \label{thm:main-ff2}
	Assume that the monotone property $P = \cup_{n \in \N} P_n$  is $1$-sided
        error strongly-testable in the $F(d)$-model. Then for any $\epsilon >0$
        there is a $r = r(\epsilon)$ such that for any $n$ there is
        a $r(\epsilon)$-set of rooted digraphs $\mathcal{H}_n$
        such that the corresponding property $P_{\mathcal{H}_n}$
        that contains the $n$-vertex digraphs that are
        ${\mathcal{H}_n}$-free, satisfies the following two conditions:

        \noindent
(a) $P_n \subseteq P_{\mathcal{H}_n} $\\
(b) $P_{\mathcal{H}_n} $ is $\epsilon/2$-close to $P_n$.
\end{thmp}

\begin{proof}
Since $P$ is strongly-testable, Theorem
\ref{thm:canonical} implies that for any $\delta$ there is a $(q,q)$-canonical
$1$-sided error $\delta$-test, $T(n,\delta)$
for $P_n$, where $q = q(\delta)$ is independent of $n$. By
definition $T(n,\delta)$ picks $q$ vertices uniformly at random, makes the
$q$-disc queries around each, and accepts or reject based only on the
 configuration of size at most $r(\delta)=q\cdot d^{q+1}$ that it
sees. 
Let $$\mathcal{C}_n(\delta)=\{C=(H,L) ~\mid ~~ \exists G \text{ with } n
\text{ vertices which is rejected by } T(n,\delta) \text{ upon seeing the configuration }
              C\}.$$
              
	$\mathcal{C}_n(\delta)$ is a well defined set of $r(\delta)$-size configurations as the
        decision of $T(n,\delta)$ depends only on the configuration it sees.
        Let $\mathcal{H}'(\delta) = \mathcal{H'}_n(\delta)=\{H ~ | C=(H,L) \in \mathcal{C}_n
        \}$. Obviously $\mathcal{H'}$ is an $r(\delta)$-set.
      For fixed $\epsilon$,   $\mathcal{H'}(\epsilon/2)$ will nearly be our required set as
        asserted in the theorem. We will show in what follows that the
        conditions (a) and (b) of the theorem hold for
        $\mathcal{H'}(\epsilon/2)$. We will then need to change it
        slightly so that every member of it is rooted while keeping
        (a) and (b). 

	\begin{claimp}\label{cl:fb-mon1}
	  For every $\delta, $	$P_n \subseteq P_{\mathcal{H'}(\delta)}$.
	\end{claimp}

	\begin{proof}
          Assume for the contrary that $G \in P_n$ but it is not
          $\mathcal{H'}(\delta)$-free. Then,  for some $V' \subseteq
          V, |V'|=|V(H)|$, $G[V']$ contains a subgraph $H \in \mathcal{H'}(\delta)$.
  Namely, there is a $1-1$ map between $V'$
          and $V(H)$ showing the isomorphism.  For simplicity we
          identify in what follows $V'$ with $V(H)$.

   We claim
   that $G$, or a subgraph of it that is obtained by removing some
   edges, has a $C$-appearance for a configuration $C = (H,L) \in
   \mathcal{C}_n$ (there is such $C=(H,L) \in \mathcal{C}_n$ by the
   definition of $\mathcal{H'}(\delta)$).  

   Indeed, we first remove the set of edges from $G$ so that $G[V']$
   is isomorphic to $H$ as an {\em induced} subgraph, resulting in a
   graph $G_1$. Now that $G_1[V']$ is isomorphic to $H$, what would
   prevent $G_1$ to have $C$-appearance with the label $L$ on the
   vertices $V'$?  The label $L$ restrict the out-degree of some
   vertices; frontier vertices must have zero degree, and developed
   vertices should have degree in $G_1$ exactly as they do in $H$ (see
   Definition \ref{def:C-free}). But, since $G_1[V']$ is isomorphic to
   $H$, removing all edges in $G_1$ that go out of $V'$ results in
   $G'$ for which the restrictions that $L$ imposes are met. So, $G'$
   has a $C$-appearance.
    
 By monotonicity
    of $P_n,$ $G' \in P_n$. Hence, (by the definition of $\mathcal{C}$)  there is positive
        probability that $T(n,\delta)$ will reject $G'$ contradicting the
        assumption that $T(n,\delta)$ is $1$-sided error test for
        $P_n$. 
\end{proof}

We note that  we crucially used here the fact  that $P$ is monotone.

        \begin{claimp}\label{cl:fb-mon2}
        For every $\delta,$  $P_{\mathcal{H'}(\delta)}$ is $\delta$-close to $P_n$.
        \end{claimp}

        \begin{proof}
          Let $G \in P_{\mathcal{H'}(\delta)}$. Then $T(n,\delta)$ accepts
          $G$ with probability $1$ by the definition of
          $\mathcal{H'}(\delta)$. Hence $G$ must be $\delta$-close to
          $P_n$ or else   $T(n,\delta)$ would have to reject it with probability
          at least $1/2$ (being an $\delta$-test for $P_n$). The other
          direction is trivial since $P_n \subseteq P_{\mathcal{H'}(\delta)}$.
        \end{proof}

        Finally, for fixed $\epsilon$ we could choose
        $\mathcal{H'}(\epsilon/2)$ to be the set guaranteed in the
        theorem,  since  by Claims \ref{cl:fb-mon1} and \ref{cl:fb-mon2} the
        conditions (a) and (b) hold for $\mathcal{H'}(\epsilon/2)$.
        However,  the theorem requires also that every $H \in
        {\mathcal{H}_n}$ is 
        rooted,  which is not guaranteed for the set $\mathcal{H'}(\epsilon/2)$.  We show in what follows that 
        $\mathcal{H'}(\epsilon/2)$ can be changed so that conditions (a) and
        $(b)$ of the theorem still hold and so that every member
        of it is rooted. 

        Let $\mathcal{\tilde{H}} = \mathcal{H'}(\epsilon/2) ~ \cup ~ \{H \in
        \mathcal{H'}(\delta)~| ~ \delta < \epsilon/2 ~ {and} ~
        |H| \leq r(\epsilon/2) \}$.
Note that $\mathcal{\tilde{H}}$ is an $r(\epsilon/2)$-set for $r()$ as
defined above. In addition, since Claim \ref{cl:fb-mon1} is true for every $\delta$, it
 follows that $P_n \subseteq
 P_{\mathcal{\tilde{H}}}$. Further, the fact that  $\mathcal{H'}(\epsilon/2)
 \subseteq \mathcal{\tilde{H}}$ implies that $P_{\mathcal{\tilde{H}}}
 \subseteq P_{\mathcal{H'}(\epsilon/2)}$, and hence by 
 Claim \ref{cl:fb-mon2}
it holds  that $P_{\mathcal{\tilde{H}}}$ is $\epsilon/2$-close to $P_n$.

It could be that there are two distinct digraphs
$H,H' \in \mathcal{\tilde{H}}$, where $H$ is a subgraph of
$H'$. For every such pair $(H,H')$ we remove $H'$ from
$\mathcal{\tilde{H}}$ so to result in the set
$\mathcal{H}=\mathcal{H}_n$ for which no member is a subgraph of
another.  This is our final set as required for the theorem. Indeed
removing $H'$ when such a pair $(H,H')$ exists does not change $P_{\mathcal{\tilde{H}}}$ at all,
and hence conditions (a) and (b) hold for $\mathcal{H}$.

  We claim that each $H \in \mathcal{H}$ is
  rooted. The argument for this also exhibits 
 the advantage of $\mathcal{H}$ in
comparison with the initial $\mathcal{H'}(\epsilon/2)$.  
 Assume for the contrary that 
        $H \in \mathcal{H}$ is not rooted, and consider the graph $G_H$
        that is composed by $n/|H|$ vertex disjoint copies of $H$.
        Proposition \ref{tester:directed-out:2} asserts that any $1$-sided
        error algorithm that needs to discover a copy of
        $H$ with constant probability makes $\Omega(\sqrt{n})$
        queries. Now, this is not a contradiction to the fact that $H$
        might be a member of $\mathcal{H'}(\epsilon/2)$ if $\epsilon/2
        > \frac{1}{d|H|},$ since the test
        $T(n, \epsilon/2)$ does not need to reject $G_H$ in this case.  However,
        this can not happen if $H$ is a member of $\mathcal{H}$:
        Indeed, since $G_H$ is 
         $\frac{1}{d|H|}$-far from $P_n$ the test  $T=T(n,\delta)$ 
          rejects $G_H$ for $\delta = \min\{\epsilon/2,
         \frac{1}{d|H|}\}$ with probability at least $1/2$. By the
         construction of $G_H$ this can be done only  by discovering a subgraph
         isomorphic to $H$ or by discovering a subgraph $H'$ of  $H$. 
The later case is ruled out since the existence of such $H'$ implies
that $H' \in \mathcal{H}$ contradicting the fact that $H \in
         \mathcal{H}$.
  The former case cannot happen as we argued that
         to discover $H$ with constant success probability takes
         $\Omega(\sqrt{n})$ queries.
\ignore{
  Recall that we have fixed $\epsilon$, and accordingly we have
  defined the set $\mathcal{H} = \mathcal{H}_n$. We now consider again
  the sets

        We first remove from $\mathcal{H}_n$ any non-minimal
        member. This does not change $P_{\mathcal{H}_n}$ at all, so
        (a) and (b) still hold by Claim \ref{cl:fb-mon1} and Claim
        \ref{cl:fb-mon2} respectively.  In what follows we refer to
        $\mathcal{H}=\mathcal{H}_n$ as the family in which every member is minimal.
        
        Assume for the contrary that 
        $H \in \mathcal{H}$ is not rooted, and consider the graph $G_H$
        that is composed by $n/|H|$ vertex disjoint copies of $H$.
        Proposition \ref{tester:directed-out:2} asserts that any $1$-sided
        error algorithm that needs to discover a copy of
        $H$ with constant probability makes $\Omega(\sqrt{n})$
        queries. This is not yet a contradiction since $G_H$ is not
         $\epsilon/2$-far from $P$ if $\epsilon >
        2d/|H|$. Hence, the test  $T=T(n,\epsilon/2)$ guaranteed for
        $P$ 
        does not have to reject $G_H$. We end the proof as follows.

       We will replace $\mathcal{H}$ with $\mathcal{\tilde{H}}$, by
       replacing $H$ with a smaller rooted subgraph $H'$ of $H$,
       and so that $P_n \subseteq P_{\mathcal{\tilde{H}}}$ still holds.  In turn, it
       will immediately follow that $P_{\mathcal{\tilde{H}}}$ is
       $\epsilon/2$-close to $P_n$, since  $P_{\mathcal{\tilde{H}}}
       \subseteq P_{\mathcal{H}}$ on account of $H'$ being a
        subgraph of $H$.

       As a result, we have removed one non-rooted subgraph from
       $\mathcal{H}$ while keeping the conditions (a) and (b) of the
       theorem. Then after a finite number of iterations the final set
       $\mathcal{\tilde{H}}$ will contain only rooted members. 

       To perform the basic operation 
Consider  $\epsilon' = \frac{1}{d|H|}$,  the corresponding
canonical test
$T'=T(n,\epsilon')$, the corresponding set of configurations, $\mathcal{C}'$, on
which $T'$ rejects,  and the corresponding forbidden
digraphs $\mathcal{H}' = \{H'~|~ C'=(H',L') \in \mathcal{C'}\}$.

Since $G_H$ is $\epsilon'$-far from $P$, it should be
rejected by $T'$ with constant probability. Hence, either $H \in
\mathcal{H}'$ or a subgraph $H'$ of $H$ is
in $\mathcal{H}'$. 
 The former case is ruled
out by the proof of Proposition \ref{tester:directed-out:2}. For the
later case, we replace $H$ by $H'$ in $\mathcal{H}$.  Namely we get
$\mathcal{H}_1 = \mathcal{H} \setminus \{H\} \cup \{H'\}$.

Evidently
$P_{{\mathcal{H}_1}} \subseteq P_{\mathcal{H}}$ (as $H'$ is a subgraph
of $H$), and further $P_n \subseteq P_{\mathcal{H}_1}$ since by
assumption, a graph containing $H'$ is not in $P$. We note that
$\mathcal{H}_1$ may not yet meet our goal, as $H'$ might also be
non-rooted. However, $H'$ is strictly smaller than $H$ (in terms of
number of edges), and hence we replace the whole argument for $H'$
replacing $H$. After a finite number of iterations, we will end
with rooted $H''$ that is a subgraph of $H$, and for which the
corresponding $\mathcal{\tilde{H}} =\mathcal{\tilde{H}}_n = \mathcal{H} \setminus \{H\} \cup
\{H''\}$ is such that  $P_n \subseteq P_{\mathcal{\tilde{H}}_n}$ as
explained above.

}
\end{proof}

\subsubsection{The `only-if' part of Theorem
\ref{thm:main-hered}}        

The following is a  restatement of  the `only if' part of Theorem
\ref{thm:main-hered}.

\begin{thmp} \label{thm:main-ff2.5}
  Assume that the hereditary
  property $P = \cup_{n \in \N} P_n$ is $1$-sided error
  strongly-testable in the $F(d)$-model. Then, for any $\epsilon >0$
  there is a  $r(\epsilon)$-set of rooted digraphs
  $\mathcal{H} = \mathcal{H}_\epsilon$ and $n^*_{\epsilon} \in \N$
  such that for
every $n > n^*_{\epsilon}$ the property $P^*_{\mathcal{H}_n}$ that contains the 
$n$-vertex  digraphs that are $\mathcal{H}$-free satisfies the
following two conditions:

\noindent
(a) $P_n \subseteq P^*_{\mathcal{H}_n} $\\
(b) $P^*_{\mathcal{H}_n} $ is $\epsilon/2$-close to $P_n$.
\end{thmp}

\begin{proof}
	Assume that $P$ is hereditary and is $1$-sided error
        strongly-testable. Theorem \ref{thm:canonical} implies
        that for any $\delta\in (0,1)$ and $n \in \N$, there is a collection of canonical-tests $\cup_{\delta \in (0,1], n\in \N} T(\delta, n), $
	where
	$T(\delta,n)$ is a $1$-sided error $(q,q)$-canonical $\delta$-test for
        $P_n$,  making at most	$q = q(\delta)$ $q$-disc queries. 

For every $\delta > 0, ~ n \in \N$, let $\mathcal{C} = \mathcal{C}(\delta,n)$ be the set of
	forbidden configurations defined by $T(\delta,n)$, namely these
	configurations on which $T(\delta,n)$ reject with some positive
	probability.

        	\begin{claimp}\label{cl:121}
  For every $\delta > 0$  and $n' > n \geq q,~$ if  $G_{n'} \in
  P_{n'}$ then $G_{n'}$ is $\mathcal{C}(\delta,n)$-free.
	\end{claimp}
	\begin{proof}
		Suppose that $G'=G_{n'} \in P_{n'}$ for $n ' > n$. If
                $G$ has a $C$-appearance for $C \in
		\mathcal{C}(\delta,n)$,  then fixing such a $C$-appearance, 
                and 
                deleting $n'-n$ vertices without
		touching the $C$-appearance in $G'$, results in 
		 a graph $G'$ on $n$ vertices that is in $P$ (as $P$ is
		hereditary). However, $G'$ has a  $C$-appearance
                causing  $T(\delta,n)$ to reject it  with positive
                probability. This contradicts the fact that 
		$T(\delta,n)$ is $1$-sided error for $P_n$.
	\end{proof}
	
Since for every fixed
	$\delta$, all tests $T(\delta,n)$ examine only configurations of
	size at most $q$ (that may depend on $\delta$  but not on
        $n$), $\mathcal{C}^*(\delta) = \cup_{n \in \N} \mathcal{C}(\delta,
        n)$ is finite. Namely, there is  some $n(\delta) \in \N$ such
        that $\mathcal{C}^*(\delta) = \cup_{n \leq n(\delta)} \mathcal{C}(\delta,
        n)$.   We conclude, by Claim \ref{cl:121}, that for
	every $n > n(\delta),$ if $G \in P_n$ then  $G$ is
        $\mathcal{C}^*(\delta)$-free. 

We now proceed with the proof of the Theorem: Fix $\epsilon$ and let
$\delta = \epsilon/2$. 
                Set
        $n^*_\epsilon = n(\delta) +dr +1$, where $r$ is the maximum
        size of a configuration in $\mathcal{C}^*(\delta)$. 
		At this point we have concluded that for every $n \geq n(\delta)$ the test
	$T(\delta,n)$ defines  the same family of forbidden
	configurations $\mathcal{C}^*(\delta) $.

        Recall that for a configuration $C=(H,L)$, if $L(v)=frontier$
        then the out-degree of
        $v \in V(H)$ is $0$. However, $G$ will have
        a $C$-appearance even if $G$ contains an {\em induced} subgraph $G'$ that is
        isomorphic to $H \cup (v,x)$, where $L(v)=frontier$ (see
        Definition \ref{def:C-free}). This motivates the following
        definition, capturing the set of possible induced graphs of
        $G$ that will cause a
        $\mathcal{C}$-appearance in $G$.

        \begin{definitionp}
          \label{def:closed}
          Let $C=(H,L)$ be a configuration in the $F(d)$-model. Then,
          $$cl(C) = \{H' =(V(H), E') ~ | ~ E(H) \subseteq E', ~ and ~\forall (v,x) \in E'
          \setminus E(H), ~ L(v)=frontier \}$$
        \end{definitionp}
Hence  $cl(C)$ consists of all digraphs $H'$ such
that if an $n$-vertex graph $G$ has a $C$-appearance on its vertices $A
\subseteq V(G)$, then $G[A]$ induces a subgraph isomorphic to $H'$
(note that 
the outdegree of a frontier vertex in $H'$ might not be zero).

Let  $\mathcal{H}= \mathcal{H}_\epsilon=  \cup_{C \in
\mathcal{C}^*(\delta)} ~ cl(C)$, and let $ P^*_{\mathcal{H}_n}$ contain
        the  $n$-vertex digraphs that are $\mathcal{H}$-free as
        induced subgraphs. By the definition of $r$,  $\mathcal{H}$ is an $r$-set. 

	\begin{claimp}\label{cl:fb-her1}
	For $n \geq n^*_{\epsilon}$, 	$~ ~ P_n \subseteq P^*_{\mathcal{H}_n}$.
	\end{claimp}
        \begin{proof}
          Assume for the contrary that  $G \in P_n$ and $G$ is not
   $P^*_{\mathcal{H}_n}$. Then for some $H \in \mathcal{H}$,  $G$
   contains an $H$-appearance as an induced subgraph on some $V_H
   \subset V(G)$. Let $C=(H,L) \in
   \mathcal{C}^*(\delta)$ be 
   the corresponding configuration for which $H \in cl(C)$. By Fact
   \ref{fact:f3} the digraph $G'$ that is obtained from $G$ by
   deleting the outgoing neighbours of $V_H$ in $G$ has a
   $C$-appearance. Let $n' = |V(G')|$.

   Note that  $n' \geq n(\delta)$.  Hence $T(\delta,n')$
   would reject $G'$ with a positive probability. But $G' \in P$ on
   account of $P$ being hereditary. This  contradicts the
   fact that $T(\delta,n')$ is a $1$-sided error for $P_{n'}$.  
 \end{proof}
 
        \begin{claimp}\label{cl:fb-her2}
    For $n \geq n^*_{\epsilon}, ~ ~       P^*_{\mathcal{H}_n}$ is $\epsilon/2$-close to $P_n$.
        \end{claimp}
        \begin{proof}
            Let $G \in P^*_{\mathcal{H}_n}$. We claim that $T(\epsilon/2,n)$
  accepts $G$ with probability $1$.  Indeed assume that  $T(\epsilon/2,n)$ 
  rejects $G$ on account of a $C$-appearance. Then by the definition of
  $\mathcal{H}_\epsilon$, $G$ would have an induced subgraph $H' \in cl(C)$ for
  some $C \in \mathcal{C}^*(\delta)$, contradicting
  the fact that $G \in P^*_{\mathcal{H}_n}$.  Hence $G$ must be $\epsilon/2$-close to
          $P_n$ as  $T$ is $\epsilon/2$-test for $P_n$. The other
          direction is trivial since $P_n \subseteq P_{\mathcal{H}}$.
        \end{proof}

	

 We have proved that the requirements (a), (b) of Theorem
 \ref{thm:main-ff2.5}  hold for the $r$-set $\mathcal{H}_\epsilon$. 
 Finally, the fact that each $H \in \mathcal{H}_\epsilon$ is  rooted 
is argued similarly as in  the proof of Theorem \ref{thm:main-ff2}
(the monotone case).
\end{proof}

\subsubsection{A few concluding remarks on Theorem \ref{thm:main-ff2} and
  monotone properties.}

 It is easy to see that if the property $\mathcal{P_{\mathcal{C}}}$  in the $F(d)$ model is
 monotone, then $\mathcal{C}$ is {\em upwards closed} in the sense that
 is defined below.

 \begin{definitionp}\label{def:upward}
A set of configurations   ${\mathcal{C}}$ is upwards-closed if for every
$C=(H,L) \in \mathcal{C}$ and $v$
being developed, adding any edge $(v,u)$   to $H$,  
while respecting the degree bound, results in a configuration $C' =
(H',L')$ that is
also in $\mathcal{C}$, where if $u \in V(H)$ then $L'
= L$, otherwise $L'(u) =
frontier$ and $L'(x)=L(x)$ for every other vertex $x$.
\end{definitionp}

\begin{factp}\label{fact:f1}
   $P_{\mathcal{C}}$ is monotone if and only if
${\mathcal{C}}$ is upwards-closed. $\qed$
\end{factp}
 An immediate conclusion from Fact \ref{fact:f1} is that  for monotone
 $P_{\mathcal{C}}$, $\mathcal{C}$ can be specified by  its minimal
 configurations. (w.r.t to Definition \ref{def:upward}). 
 Next, we generalize Theorem \ref{thm:main-ff2}, moving beyond the
 scope of monotone properties. Towards this end we use the following.

\begin{definitionp} [Rooted Configuration]
	A  configuration $C = (H,L)$, where $H$ is a digraph
        and $L$ is a label function, is {\rm rooted} if $H$ is
        rooted. 
      \end{definitionp}

A conclusion from the proof of Theorem \ref{thm:main-ff2} is that
$\mathcal{C}_n$ as defined in the proof is upwards closed and every
minimal (with respect to Definition \ref{def:upward}) configuration in
it is rooted. 
However, more can be said:  The
following theorem follows directly from the arguments above {\em for any
  digraph property}, where we say that a set of configuration $\mathcal{C}$ is an $r$-set if for
every $C = (H,L) \in \mathcal{C}, ~ |V(H)| \leq r$.

\begin{thmp}\label{thm:general}
Assume that the digraph property $P = \cup_{n \in \N} P_n$  is $1$-sided
        error strongly-testable in the $F(d)$-model. Then, for any $\epsilon >0$
        there is a $r = r(\epsilon)$ such that for any $n$ there is
        $r(\epsilon)$-set $\mathcal{C}_n$ of  
        configurations such that every minimal configuration in
        $\mathcal{C}_n$ is rooted and,
\noindent
(a) $P_n \subseteq P_{\mathcal{C}_n}$ and 
(b) $P_{\mathcal{C}_n} $ is  $\epsilon/2$-close to $P_n$.
\end{thmp}

The proof is essentially identical to the proof of Theorem \ref{thm:main-ff2},
in which we replace subgraphs by configurations and leave out the
parts dealing with monotonicity. 

\section{Strongly-Testable properties that are non-monotone neither hereditary}\label{sec:5}
There are $1$-sided-error strongly-testable properties in the
$F(d)$-model (and in all other models too) that are not monotone,
neither are hereditary.  Consider e.g., the $F(d)$-model and the
property $P$ of not having a vertex of out-degree $d-1$. This property
is not trivial e.g., the graph that contains $n/d$ vertex disjoint
directed $(d-1)$-stars  is $\frac{1}{d}$-far from the property.  Moreover
$P$ is non-monotone and not hereditary.  But $P$ is strongly-testable as
if $G$ is $\epsilon$-far from $P$ then $G$ contains at least
$\epsilon n$ vertices of degree $d-1$. Indeed, it can be defined by
one forbidden rooted configuration, hence consistent with Theorem
\ref{thm:general}.

A more interesting property that is $1$-sided error strongly-testable while not
monotone nor hereditary is the following property {\em
	RV} (for a ``reachable vertex''). Differently from not having a degree $d-1$ vertex, the
property $\Sink$ is not expressible by a finite collection of forbidden
configurations at all. Rather, it is close to such (for any $\epsilon$).

For a digraph $G=(V,E)$ a vertex $s \in V$ is
called   ``reachable-by-all''  if  there is a directed path
from each vertex in $G$ to $s$. 
Note that $G$ may have many such vertices, in particular, if $G$ is
strongly connected then every vertex  is reachable-by-all. Let $\Sink$ be the digraph property of having a
vertex that is reachable-by-all. The
property 
$\Sink$  is
not trivial, as e.g., a directed matching is far from $\Sink$.

\begin{thm} \label{thm:sink}
	The property $\Sink$ is $1$-sided error strongly-testable in the $F(d)$-model.
\end{thm}

\begin{proof}
	The following test $T$ is a $1$-sided error   $\epsilon$-test for $\Sink$ making
	$\frac{1}{d^2\epsilon^2} \cdot d^{O(1/(d\epsilon))}$ 
        queries. The basic idea is very similar to the test (and
        proof) for testing connectivity in \cite{GR02}.
	
	\noindent
	{\bf Test for $\Sink$, $T(\epsilon)$, for $\epsilon < 1/d$:}
	
	\begin{enumerate}
		\item Choose a multiset of vertices $B \subseteq V(G)$
		by choosing independently a vertex $v \in V(G)$
                uniformly at random,                 for $b= \frac{200}{d^2\epsilon^2}$ times.
                Let $B=\{v_1, \ldots ,v_b \}$ the vertices thus chosen.
		\item For $i=1$ to $b$: query the disc
                  $D(v_i,\frac{2}{d\epsilon})$ around $v_i$, and let $S_i$ be the
		set of vertices that is discovered (including $v_i$). 
		\item If there are distinct  $i,j$ such that $\Gamma^+(S_i) = \Gamma^+(S_j)
		= \emptyset$ and $S_j \cap S_i = \emptyset$ reject, otherwise accept.
	\end{enumerate}
	\begin{claim}
		$T(\epsilon)$ never reject a digraph in $\Sink$.
	\end{claim}
	\begin{proof}
		Let $G$ have a vertex $a$ that is
                reachable-by-all. Then for any $v$ that
		is queried,  there is a path from $v$ to $a$. Therefore,
                for every $i$, either $a$ is in
		$D_{v_i} = D(v_i,2/(d\epsilon)~)$, or  there is a path from $v_i$ to $a$
		that stretches outside $D(v_i,\frac{2}{d\epsilon})$
                implying that $\Gamma^+(S_i) \neq
                \emptyset$. Therefore, for every  $v_i$ and 
		$v_j,$ $\Gamma^+(S_i)=\Gamma^+(S_j)= \emptyset$ holds
                only when $a \in S_i \cap S_j$.
	\end{proof}

	\begin{claim}
		Let $\epsilon < 1/d$ and $G$ be $\epsilon$-far from $\Sink$, then $T(\epsilon)$ rejects $G$ with
		probability at least $1/2$.
	\end{claim}
	\begin{proof}
		Let $G$ be $\epsilon$-far from $\Sink$ and $SC(G) = (A,F)$ be the DAG
		of the strongly
		connected components of $G$. We first claim that $SC(G)$ contains at
		least $\epsilon
		dn$ components $c \in A$ for which $\Gamma^+(c)=\emptyset$. Indeed, let $c_1, \ldots ,c_k$ be the strongly connected
		components of $G$ for which $\Gamma^+(c_i) =
		\emptyset$. To see that $k \geq \epsilon dn$ note that
                by changing at most $k-1$ edges
		(one per $c_i$, connecting it to $c_{i+1}$),  $G$ will have a
		a vertex that is reachable-by-all in $c_k$.  
		
                This implies that there are at least $\epsilon dn/2$
                components $c\in A$, of size at most $2/(d\epsilon)$,
                for which $\Gamma^+(c) = \emptyset$.  We denote this
                set of components by $A^*$ and the vertices in $A^*$
                by $V^*$. It follows that $|V^*| \geq \epsilon dn/2$,
                and hence, with high probability sampling
                $b= \frac{200}{d^2\epsilon^2}$ vertices finds  two vertices in two distinct
                components in $A^*$. Scanning the
                $\frac{2}{d\epsilon}$-disc around two such vertices
                will cause the test to reject.  \ignore{
                  We conclude that picking a random vertex $v$, it
                  will be in $V^*$ with probability at least
                  $\epsilon d/2$.  In turn, the expected number of
                  vertices in $B \cap V^*$ is at least
                  $b \cdot \epsilon d/2$. Hence, by Chernoff, with
                  extremely small probability
                  $|B \cap V^*| \leq \epsilon bd/100 =
                  2/d\epsilon$. In addition with extremely large
                  probability all the $b$ vertices chosen for $B$ are
                  distinct. Assuming that the vertices in $B$ are
                  distinct and that $|B \cap V^*| > 2/(d\epsilon)$,
                  then there are two distinct vertices $v,v'$ that
                  reside in distinct components from $A$ (as each
                  component in $A$ has size bounded by
                  $2/(d\epsilon)$).  Two such vertices $v,v'$ will
                  cause the test to reject.}
	\end{proof}

	Finally, the query complexity is clearly $b \cdot \max_i |S_i|$ which is as stated.
\end{proof}

We note that $\Sink$ cannot be defined by any $r$-set of
forbidden configurations, for $r$ that is independent of $n$. To see this consider a digraph that is
composed of two vertex disjoint simple di-cycles of length $n/2$
each. Such a graph is not in $\Sink$ but every configuration of it of
size at most $n/4$ is shared by the digraph that is composed of one
single directed cycle, which is in $\Sink$.  The property
$P_{\mathcal{C}_\epsilon}$ that is actually being tested by a
$1$-sided error test for $\Sink$ is defined by the set
$\mathcal{C}_\epsilon$ in which every configuration is a pair of
vertex-disjoint discs, of the appropriate size, with no outgoing
edges. The property $\Sink$ is a subset of $P_{\mathcal{C}_\epsilon}$
for every $\epsilon$, but the size of $\mathcal{C}_\epsilon$ while
finite for every $\epsilon$, is not bounded when $\epsilon$ tends to
$0$.


\section{The $FB(d)$-model and the undirected bounded-degree graph model}\label{sec:FB}

As already mentioned, the undirected bounded-degree graph model can be
viewed as a submodel of the $FB(d)$-model. Hence, we state the results
only for the $FB(d)$-model. The results are very similar to these for the
$F(d)$-model, except that the restriction that the forbidden members
are rooted
is not needed. In addition, the  test for being free of a 
finite family of forbidden {\em induced} graphs is similar to the
monotone case due to the  bound on incoming degree (this will
further explained in the relevant place below). We define here the appropriate notions and state
the appropriate theorems. We give proofs only where they are 
significantly different from these for the $F(d)$ model.

We start with the relevant notions, analogous to these seen for the
$F(d)$-model.
The first notion,  which is non-standard due to the type of queries that
is available, is that of $r$-disc.

\begin{definition} [$r$-disc, $FB(d)$-model]\label{def:fb-disc}
	Let $r$ be an integer and  $v \in V(G)$.  $\tilde{D}(v,r)$
        denotes the ``$r$-disc''  for the $FB(d)$-model, and is defined recursively as follows:

        $\tilde{D}(v,1) = \{v\} \cup \Gamma^+(v) \cup \Gamma^-(v)$.

       For $r \geq 2,~$         $\tilde{D}(v,r) = \cup_{u \in \tilde{D}(v,1)} \tilde{D}(u,r-1)$
\end{definition}

        That is, $\tilde{D}(v,r)$ contains all vertices that are reachable
        from $v$ by path of length at most $r$ that is composed of
        edges that may be traversed in the wrong direction. The point
        being that the $FB(d)$-model allows for such traversal.
        With Definition \ref{def:fb-disc}, an $r$-disc query in the $FB(d)$-model
        is defined exactly as in the $F(d)$-model, where $r$-disc are
        the corresponding one.
         $r$-disc queries generalize basic neighbourhood queries  as in the $F(d)$
model, and with the same complexity overhead.

For a family of (di)graphs $\mathcal{H}$, the definitions of being
$\mathcal{H}$-free as subgraphs, or as induced subgraphs are 
extended naturally with no  alterations (as these are model-independent definitions). 
 But configurations for
the $FB(d)$ model are defined slightly differently; a configuration is
defined as for the $F(d)$ model, with the extra restriction that the
degree bound holds for both in-degree and out-degree. In addition,
frontier vertices may have non-zero out-degree.
Being $C=(H,L)$-free, for a
configuration $C$, is defined as follows.

\begin{definition} [$C$-Free, $FB(d)$-model]
	Let $C=(H,L)$ be a configuration and $G=(V,E)$ a
        $d$-bounded degree digraph in the $FB(d)$ model.  Let $V'
        \subseteq V$. We say that $G[V']$ is a  $C$-appearance if
 there is a bijection $\phi : V(H) \rightarrow
        V'$ such that $\forall v,u \in V(H)$ and $L(v)=\rm{developed}$,
        \begin{center}
          $(v,u) \in E(H) \leftrightarrow
	(\phi(v),\phi(u)) \in E~ ~ $
 and $~ ~  (u,v) \in E(H) \leftrightarrow (\phi(u), \phi(v)) \in E$. 
      \end{center}
Further, for every developed $v,$ if $(\phi(v), x) \in E$
 or $(x,\phi(v)) \in E$ 
 then $\exists	u\in V(H), \phi(u)=x$.
 
	We say that $G$ is $C$-free if $G$ has no $C$-appearance.  
\end{definition}

Finally,  the fact that every strongly-testable property is testable by a
canonical tester is also identically the same.
We get the following analog of Theorem \ref{thm:main-mon}.

\begin{thm}  \label{thm:main-fb}
	A monotone digraph property $P = \cup_n P_n$ is $1$-sided error strongly-testable in
        the $FB(d)$-model if and only if for every $\epsilon> 0$ there is a $r = r(\epsilon)$ such that for any $n$ there is a $r$-set of 
         digraphs $\mathcal{H}_n$ for which the following two
         conditions hold  (a) $P_n \subseteq P_{\mathcal{H}_n} $ and (b): 
	$P_{\mathcal{H}_n} $ is $\epsilon/2$-close to $P_n$.  $\qed$
\end{thm}
The proof is mostly identical to the corresponding proofs for
the $F(d)$ model and is omitted.
Note that we do not require here that the forbidden digraphs are 
rooted. This is not needed anymore, due to the stronger query-type.
The analogous theorem for hereditary properties is.
\begin{thm}\label{thm:main-fb0.5}
	An hereditary digraph property $P = \cup_n P_n$ is $1$-sided error strongly-testable in
        the $FB(d)$-model if and only if for every $\epsilon> 0$ there
        is a $r = r(\epsilon)$, $n^*_{\epsilon} \in \N$ and a    
        $r$- set of  digraphs $\mathcal{H}$,   for which
        the following conditions hold: for
        every $n > n^*_{\epsilon}~$ 
        (a) $P_n \subseteq P^*_{\mathcal{H}}$, and (b) $P^*_{\mathcal{H}} $ is  $\epsilon/2$-close to $P$.
\end{thm}

\begin{proof}
  The proof of the `only-if' part is identical to that of Theorem
  \ref{thm:main-ff2.5} without the restriction (and complication) of
  being rooted. 

For the `if' part, the analog of Theorem \ref{thm:main-ff1.5} holds
with a simpler proof. The proof starts identically, with $S$ being a
maximal set of induced subgraphs of $G$, each being an
$\mathcal{H}$-appearance (with no restrictions on roots). Then,  
deleting all edges adjacent to vertices appearing in $S$ results in
$G$ becoming $\mathcal{H}$-free.
(In the $F(d)$-model, we could not afford deleting all edges adjacent to
$S$ as this could be a large set while $S$ is small, and we had to
resort to sampling a random edge.  Here, due to the in-degree bound, if
 $G$  is $\epsilon$-far from $\mathcal{H}$-free
then $|S| \geq \epsilon n/4$ (as in the first case of Proposition 
\ref{thm:main-ff0.5}).)

The rest of the ``only if'' direction follows from the analog of
Theorem \ref{thm:main-ff1.5}, which is identically stated for the
$FB(d)$ model, leaving out the restriction of that members of $\mathcal{H}$ are
rooted. 
\end{proof}


\section{ Two application of the characterization}\label{sec:55}

A characterization is more useful when apart of giving some 
structural insight to a feature, it also allows to simply conclude the
existence or lack of a property using the characterization and without
going into the theory behind it. Here we show two  applications of
our characterization for proving known results. The first is to show
that the monotone (and hereditary)  property of
being $2$-colourable is not strongly-testable (proved  in
\cite{GR02}). The second is that the monotone (and also hereditary) property of being $k$-star-free as a {\em minor} is strongly
testable (done as a part of proving other results in
\cite{gold-tree-minor}).  The discussion below is done with respect to
the undirected $d$-bounded degree model.

\subsection{$k$-colorability}
It is known that $k$ colorability is
not strongly-testable (even by $2$-sided error tests) for bounded-degree graphs for $k \geq 2$ \cite{GR02}. 
Here we reprove the fact without getting into property testing  at all. We use the analogous theorem  of Theorem
\ref{thm:main-fb} for the undirected model.

Indeed, since $2$-colorability is monotone,  if it were
strongly-testable, then the analog of of Theorem
\ref{thm:main-fb} for the undirected bounded-degree model 
would imply that  there is a $r=r(\epsilon)$ and a $r$-set $\mathcal{H}_\epsilon$ such that
the corresponding conditions (a) and (b) hold. 
Namely, there should be a $r$-set $\mathcal{H}$ of graphs such that:
(a) $2$-colorability must be a
subset of a property $P_{\mathcal{H}}$, and  (b) that  $P_{\mathcal{H}}$ should be $\epsilon$-close to being
$2$-colourable.

Assume that  $\mathcal{H} = \mathcal{H}_\epsilon$
is such a set. 
By (a) every $H \in \mathcal{H}$ is not
$2$-colourable.  Further, $\mathcal{H}$ must contain {\em all} non-$2$-colourable
graphs up to size $d/\epsilon$ (otherwise if a non-$2$-colourable
graph $H_0$ of size
smaller than $d/\epsilon$   is not in
$\mathcal{H}$, then the graph that is composed of $nd/|H_0|$ disjoint
copies of $H_0$  is $\epsilon$-far from $2$-colorability but is in
$P_{\mathcal{H}}$).

Let $d$ be large enough, $\epsilon$ small enough, and take any
good $d$-regular Ramanujan expander (or random $d$-bounded degree graph with no
short cycles). Such a graph is locally a tree, and 
hence $\mathcal{H}$-free. However, it is $\epsilon$-far from being
$2$-colourable, as by the expander mixing lemma, 
any bipartition of the vertex set has many more than
$\epsilon dn$ edges with both ends in one of the parts. We omit further details. 

\subsection{Being $k$-star-free}
The property of $d$-bounded degree undirected graphs of being $k$-star
free as minors is a monotone and hereditary property. It is a simple
instance of the more complex property of being $\mathcal{H}$-minor
free, for a fixed given set of graphs $\mathcal{H}$.  It is known and
obvious that for  arbitrary $\mathcal{H}$, the property of being
$\mathcal{H}$-minor free is not strongly-testable by $1$-sided error
algorithms, as even acyclicity (namely not having a triangle minor) is not $1$-sided error strongly
testable for $d \geq 3$ \cite{GR02}.
However, for $\mathcal{H}$ being a fixed collection of trees, the
property of being $\mathcal{H}$-free is strongly-testable as was shown
in \cite{gold-tree-minor}. A first (and relatively easy step) in the result of
\cite{gold-tree-minor} is when the only member of $\mathcal{H}$ is the $k$-star (for constant fixed
$k$).

The property $P$ of being $k$-star free as a minor is a monotone property.
 We show that being $k$-star free as a minor is $1$-sided error
 strongly-testable for the undirected
$d$-bounded degree graph model using Theorem \ref{thm:main-fb}.
 Indeed, all we need to show
(for any $\epsilon > 0$)  
is an $r(\epsilon)$-set $\mathcal{H}$ such that 
following holds: (a) $P_n \subseteq
P_{\mathcal{H}_n}$ and  (b)  that $P_{\mathcal{H}_n}$ is $\epsilon$-close to
$P_n$. Here $P_{\mathcal{H}_n}$ contains the  $n$-vertex graphs in
$P_{\mathcal{H}}$. 

  We set $\mathcal{H}$ to contain all graphs of size at most
  $s=\frac{k}{\epsilon} + kd$ that 
  contain a $k$-star  as a minor. It is obvious from the definition
  that $P_n \subseteq P_{\mathcal{H}_n}$.

  Let $G \in P_{\mathcal{H}_n}$. We note that for any $S \subseteq
  V(G)$ such that $G[S]$ is connected and $|S| \leq s-k$, the edge cut $(S,\bar{S}) = \{(u,v) \in E(G)~|~ u
  \in S, ~ v \notin S\}$ has size at most $kd$. This is true as
  otherwise contracting $G[S]$ to a single point exhibits a $k$-star
  in the subgraph $G[S \cup \Gamma(S)]$ that is of size at most $s+kd$.

  Hence, it follows that we can decompose $G$ by iteratively choosing
  a vertex $v$ in a large enough component, and removing any connected
  subgraph of size $s-k$ containing $v$.  This will result in
  components of size at most $s-k$, while removing at most $kd$ edges
  at each iteration. Thus in total, removing at most
  $\frac{n}{s-k}\cdot dk$
  edges we get a graph $G'$ that is a subgraph of $G$, and in which
  every component is of size at most $s-k$. It follows that $G' \in
  P_n$ by definition, and since we have removed at most
  $\frac{ndk}{s-k}  \leq \epsilon dn$ it implies that $G$ is
  $\epsilon$-close to $P_n$.

\section{Concluding Discussion}\label{sec:concl}

 Let $\mathcal{C}$ be
a finite set of configurations (in any of the models discussed
above). The property of being $\mathcal{C}$-free is very natural in the
context of bounded-degree (di)graphs. In particular, all monotone and
all hereditary
properties are instances of such properties. Hence, being free of
$\mathcal{C}$ is a collection of properties worth studying (and not
only in the context of property testing). 

We have characterized the monotone and hereditary (di)graph properties
that are $1$-sided error strongly-testable in all the corresponding bounded-degree (di)graph
models.   Theorem \ref{thm:general}  states that every property that is
$1$-sided error strongly-testable in the $F(d)$-model (and the analogous statements
for the other models) is defined by a finite
collection of {\em forbidden configurations} with properties (a)
and (b) as in the theorem. It could be that these are exactly the
properties that are $1$-sided error strongly-testable {\em regardless of being
monotone or hereditary}.  The
problem with extending it to a characterization arises  for the
analog of 
Proposition \ref{theorem:directed-out:1}. We do not know that for a finite
 set of rooted  {\em configurations} $\mathcal{C}$, $P=P_{\mathcal{C}}$ is
strongly-testable. It could be that
for $G$ that is $\epsilon$-far from $P$, $G$ has only a small
number of appearances 
of forbidden configurations and any way of ``correcting'' these
appearances creates new appearances. We do know this e.g., for  the
$F(d)$-model if the set of forbidden configurations are
degree bounded\footnote{Forbidden configurations of bounded degree
  $d-1$ graphs are easy to `correct' by adding edges so to create vertices of
  degree $d$. Hence in this case, if a graph is far from the
  property, then it has many vertices in forbidden configurations.} by $d-1$, but not for the general case. 

Finally, in the very simple case of the $FB(1)$-model, and hence the {\em undirected}
$2$-degree bounded model too, the inverse of
Theorem \ref{thm:general} does work.
We prove, in this case, that if a graph is far from being
$\mathcal{C}$-free then it has many
$\mathcal{C}$-appearances. This conclusion turns out to be not
entirely trivial,  although the family of $2$-degree bounded
graphs is very simple\footnote{For these models every
  (di)graph property is strongly-testable by
  the results of \cite{NS13}. However, not all properties are $1$-sided
  error strongly-testable. E.g., consider the property ``having exactly $n/2$
  edges'' for which we do not have small witnesses for being far.}. The argument requires some global considerations beyond
these used for monotone properties and appear in the Appendix.

\vspace{0.7cm}
{\bf Acknowledgment:} we thank  Oded
  Goldreich for the extensive work he has done in order to improve the presentation of this
  paper. 
\bibliographystyle{plain}

\newcommand{\Proc}{Proceedings of the~}
\newcommand{\ALENEX}{Workshop on Algorithm Engineering and Experiments (ALENEX)}
\newcommand{\BEATCS}{Bulletin of the European Association for Theoretical Computer Science (BEATCS)}
\newcommand{\CCCG}{Canadian Conference on Computational Geometry (CCCG)}
\newcommand{\CIAC}{Italian Conference on Algorithms and Complexity (CIAC)}
\newcommand{\COCOON}{Annual International Computing Combinatorics Conference (COCOON)}
\newcommand{\COLT}{Annual Conference on Learning Theory (COLT)}
\newcommand{\COMPGEOM}{Annual ACM Symposium on Computational Geometry}
\newcommand{\DCGEOM}{Discrete \& Computational Geometry}
\newcommand{\DISC}{International Symposium on Distributed Computing (DISC)}
\newcommand{\ECCC}{Electronic Colloquium on Computational Complexity (ECCC)}
\newcommand{\ESA}{Annual European Symposium on Algorithms (ESA)}
\newcommand{\FOCS}{IEEE Symposium on Foundations of Computer Science (FOCS)}
\newcommand{\FSTTCS}{Foundations of Software Technology and Theoretical Computer Science (FSTTCS)}
\newcommand{\ICALP}{Annual International Colloquium on Automata, Languages and Programming (ICALP)}
\newcommand{\ICCCN}{IEEE International Conference on Computer Communications and Networks (ICCCN)}
\newcommand{\ICDCS}{International Conference on Distributed Computing Systems (ICDCS)}
\newcommand{\IJCGA}{International Journal of Computational Geometry and Applications}
\newcommand{\INFOCOM}{IEEE INFOCOM}
\newcommand{\IPCO}{International Integer Programming and Combinatorial Optimization Conference (IPCO)}
\newcommand{\ISAAC}{International Symposium on Algorithms and Computation (ISAAC)}
\newcommand{\ISTCS}{Israel Symposium on Theory of Computing and Systems (ISTCS)}
\newcommand{\JACM}{Journal of the ACM}
\newcommand{\LNCS}{Lecture Notes in Computer Science}
\newcommand{\PODS}{ACM SIGMOD Symposium on Principles of Database Systems (PODS)}
\newcommand{\SICOMP}{SIAM Journal on Computing}
\newcommand{\SODA}{Annual ACM-SIAM Symposium on Discrete Algorithms (SODA)}
\newcommand{\SPAA}{Annual ACM Symposium on Parallel Algorithms and Architectures (SPAA)}
\newcommand{\STACS}{Annual Symposium on Theoretical Aspects of Computer Science (STACS)}
\newcommand{\STOC}{Annual ACM Symposium on Theory of Computing (STOC)}
\newcommand{\SWAT}{Scandinavian Workshop on Algorithm Theory (SWAT)}
\newcommand{\UAI}{Conference on Uncertainty in Artificial Intelligence (UAI)}
\newcommand{\WADS}{Workshop on Algorithms and Data Structures (WADS)}

\newpage
\appendix{\huge{Appendix}}
\section{``Removal Lemma'' - the case of $2$ bounded degree undirected
  graphs and the $F(1)$-model }

The standard removal lemma in our context would be that if a graph $G$
has ``small'' number of $\mathcal{C}$-appearances than it can be made
$\mathcal{C}$-free by removing and inserting a ``small'' number of
edges. Here $\mathcal{C}$ is a collection of configurations rather
than just forbidden subgraphs.  We do not know if such a lemma is
correct for $d$-bounded degree graphs and $d \geq 3$.  We prove the
following for $2$-bounded degree graphs. It is a very simple case, but
already exhibits why simple local considerations might not be enough.

\begin{lemma}\label{lem:d=1}
Let $\mathcal{C}$ be a $k$-set of forbidden configurations in the
$2$-bounded degree model for undirected graphs. 
Let $P_n$ be the property that contains the  $n$-vertex
$\mathcal{C}$-free $2$-bounded degree  graphs. For $\epsilon <
\frac{1}{4k}$ if $P_n  \neq
\emptyset$ and 
 $G$ is $\epsilon$-far from $P_n$ then $G$ contains $\epsilon^2 n/k$
 vertices in $\mathcal{C}$-appearances. 
  \end{lemma}
Before we present the proof we point  why local consideration as in
the proof of Proposition \ref{theorem:directed-out:1} are not
sufficient.  Let
$\mathcal{C}$  contain two forbidden configurations: a singleton and a
path of length $2$ where the middle vertex is Developed and the two
endpoints are Frontier. 
Consider the property $P$ of being $\mathcal{C}$-free.
If $G$ is in $P$ then $G$ is a
perfect matching and hence the property is not trivial. However, for odd $n$ $P_n= \emptyset$
 and the existence of a single $C$-appearances can not be corrected
at all. 
  \begin{proof}
A configuration $C \in \mathcal{C}$ may be disconnected and composed
of several components. For simplicity we prove the lemma for the case
that for every $C=(H,L) \in \mathcal{C}$, $H$ is connected. The proof for the general
case is more complicated but uses the same ideas.

By assumption $\mathcal{C}$ is a collection cycles
and paths. We may assume that all vertices of degree $2$ in any $C \in
\mathcal{C}$  are
Developed  as
$d \leq 2$. There are $3$ possible types of  paths in $\mathcal{C}$: A path with both ends
Developed,  Both ends Frontier, and a Frontier and Developed ends. We
call such paths $DD$, $FF$ and $FD$ paths respectively. We consider
the zero length path containing  a
single isolated vertex as a $DD$ path. 

Let $P_n$ be the property that contains the $n$-vertex graphs that are
$\mathcal{C}$-free, and assume that $P_n \neq \emptyset$.
  Let $G$ on $n$ vertices be $\epsilon$-far from $P_n$.
  Let $c$ be a component of $G$. If $|V(c)| > k/\epsilon$,    $c$ is
called `large' and otherwise it is called  `small'.

{\bf (i)}  Assume first that $\mathcal{C}$ contains no $FF$ path.

A large component $c$ of $G$ may have a $C$-appearance for some $C \in \mathcal{C}$
only if $C$ is a $FD$ path and $c$ is a path. Then by adding the edge
between the endpoints of the path $c$ it will not have a
$\mathcal{C}$-appearances.

Since there are at less than $\epsilon n/k$ large $C$'s,  all
relevant appearances are corrected by changing at most $\epsilon n/k$
edges. Let $S$ be the set of all small components of $G$ that have a 
$\mathcal{C}$-appearance. Let $|S| = \ell.$
We conclude that $\ell \geq \epsilon n/2$  as other wise we can change
$\cup_{c \in S} c$ into a
unique cycle  using at most $2\ell$ edge additions (and if $\ell \leq
k$ we can further make this cycle to be of size at least $k+1$ using
some extra $4$ edge changes) and get a graph that is
$\mathcal{C}$-free.

We conclude that there are at least $\ell \geq \epsilon n/2$ vertex
disjoint $\mathcal{C}$-appearances in $G$, which implies the lemma.

{\bf (ii)} Assume that $\mathcal{C}$ contains a $FF$ path and let $r ~ (\leq
k)$ be the
length of the smallest such path.  

{\bf (ii).1}
Assume first that $\mathcal{C}$ contain no singleton.

If a component $c$ in $G$ contains an $r$-length path then every
vertex in it is in an $r$-length $FF$ path. Let $S$ be the set of 
vertices in an $r$-length path. Then either $|S|  \geq \epsilon^2 n/k$
and we are done, or $|S| <
\epsilon^2n/k$ and we 
can delete all edges adjacent to vertices in $S$, (at most
$2|S|$ edges) to obtain a graph
that is free of $r$-size $FF$ paths, and we are back in the previous
case. 

{\bf (ii).2} 
Assume now that $\mathcal{C}$ contains a singleton (which makes the
corrections in the proofs above impossible, and as shown by the example
before the proof, corrections can not be done locally).
 
Assume first that for some $\ell \in \N$ there are graphs $G_\ell$ and
$G_{\ell+1}$ on $\ell$ and $\ell+1$ vertices respectively, and such
that both graphs are  $\mathcal{C}$-free.  Let $\ell$ be the
smallest such integer.
By a basic fact in number theory (Frobenius coin problem), there exists $n_0$ such that for
every $m > n_0,$ $~m$ can be written as $m = a\ell + b(\ell+1)$ for
some $a,b, \geq 0$. We conclude that for any $m \geq n_0$ there is a
graph $G_m$ on $m$ vertices that is composed of $a$ copies of $G_\ell$
and $b$ copies of $G_{\ell+1}$ and that is a $\mathcal{C}$-free.

Let $G$ on $n$ vertices ($n >> \ell$) that is $\epsilon$-far from being
$\mathcal{C}$-free.
 Let $S$ be the set of
vertices in large components of $G$. Note that every vertex in every
large component is in a $FF$-forbidden path (except possibly two). 
Then, if  $|S| \geq \e^2 n/k$ we are done.

Otherwise, let $S_1$ be the set of singletons in $G$.  Let  $A$ be the set
of vertices in small
components that contain a ${\mathcal{C}}$-appearance. 
If $|S_1|  \geq \epsilon^2 n/k$, or $|A|\geq
\epsilon^2n/k$ then we are done. 
 Otherwise,  let $V_1 = S \cup S_1 \cup A$. Note that $m=|V_1| \leq 
3\epsilon^2n/k$. We  assume here that $m \geq n_0$ or otherwise we add
some arbitrary $n_0-m$ additional vertices to $V_1$. We now form a
subgraph $G_m$ on $V_1$ that is $\mathcal{C}$-free. Note  that such
$G_m$ exists by our assumptions on $m$. Hence by changing at most $2m
< 6 \epsilon^2 n/k < 2 \epsilon n$  edges (for $n$ large enough), 
we have made $G'$ be ${\mathcal{C}}$-free in contradiction with the
assumption that $G$ is $\epsilon$-far from $P_n$.

Assume now that there is no $\ell$ for which $G_\ell, ~ G_{\ell+1}$
are $\mathcal{C}$-free. In this case it follows (by the same coin problem of Frobenius)
that all $G_\ell$ that are $\mathcal{C}$-free have number of vertices
that is congruent to $0$ mod some $\alpha >0$. More over, there is
some fixed $\ell\equiv 0 (\alpha) $ and $G_1, G_2$ on
$\ell, \ell+\alpha$ vertices correspondingly, such that both $G_1,
G_2$ are $\mathcal{C}$-free.  This implies that for any $m \equiv 0
(\alpha)$ that is large enough, there is a graph $G_m$ on $m$ vertices
that is $\mathcal{C}$-free.

Let $S,S_1, A, V_1$
as before. We proceed in a similar way: Either $S$ or $S_1$ or $A$  is large enough and
we are done. Otherwise, we make a subgraph $G_m$ on $V_1$ to be
$\mathcal{C}$-free, to result in a graph on $V$ that is $\mathcal{C}$-free,
contradicting the assumption that$G$ is $\epsilon$-far from being
$\mathcal{C}$-free.

The important point here is that since $P_n \neq
\emptyset$ it follows by the discussion above that $n\equiv 0
(\alpha)$. And since $G[V \setminus V_1]]$ is $\mathcal{C}$-free, it follows that $m \equiv
0 (\alpha)$. Hence a $G_m$ on $m$ vertices that is $\mathcal{C}$-free
exists.
  \end{proof}

\end{document}